\documentclass[12pt]{article}

\usepackage[top=1in, bottom=1.25in, left=1in, right=1in]{geometry} 

\usepackage[utf8]{inputenc} 
\usepackage[T1]{fontenc}    
\usepackage{hyperref}       
\usepackage{url}            
\usepackage{booktabs}       
\usepackage{amsfonts}       
\usepackage{nicefrac}       
\usepackage{microtype}      
\usepackage[dvipsnames]{xcolor}

\usepackage{natbib}

\usepackage[many]{tcolorbox}
\usepackage{authblk}

\usepackage{amsthm}

\usepackage{bm}

\usepackage{bbm}
\usepackage{mathtools}
\usepackage{amsmath}
\usepackage{thm-restate}

\usepackage{enumitem}

\usepackage{complexity}

\let\R\relax
\let\E\relax

\usepackage{MnSymbol}

\usepackage{cleveref}

\crefname{assumption}{assumption}{assumptions}
\crefname{property}{property}{properties}

\newtheorem{theorem}{Theorem}[section]
\newtheorem{proposition}[theorem]{Proposition}
\newtheorem{property}[theorem]{Property}
\newtheorem{lemma}[theorem]{Lemma}

\newtheorem{corollary}[theorem]{Corollary}

\theoremstyle{definition}
\newtheorem{definition}[theorem]{Definition}
\newtheorem{assumption}[theorem]{Assumption}
\theoremstyle{remark}
\newtheorem{remark}[theorem]{Remark}

\newcommand*{\N}{{\mathbb{N}}}

\let\R\relax
\newcommand*{\R}{{\mathbb{R}}}
\newcommand*{\E}{{\mathbb{E}}}

\newcommand{\cntrl}{\mathsf{cntrl}_{\cG}}
\newcommand{\argcntrl}{\mathsf{cntrl}}

\let\cS\relax

\newcommand*{\cX}{{\mathcal{X}}}

\newcommand*{\cZ}{{\mathcal{Z}}}

\newcommand*{\cG}{{\mathcal{G}}}
\newcommand*{\cS}{{\mathcal{S}}}
\newcommand*{\cN}{{\mathcal{N}}}
\newcommand*{\cA}{{\mathcal{A}}}

\newcommand{\defeq}{\coloneqq}

\newcommand{\eqgap}{\textsc{EqGap}}

\newcommand{\Pistar}{\Pi^\star}

\newcommand{\noisyF}{F^{\delta, \rho}}

\newcommand{\ogd}{\mathtt{OGD}}

\newcommand{\vrho}{\vec{\rho}}
\newcommand{\va}{\vec{a}}
\newcommand{\vpi}{\vec{\pi}}
\newcommand{\vx}{\vec{x}}
\newcommand{\vmu}{\vec{\mu}}
\newcommand{\hvx}{\hat{\vec{x}}}
\newcommand{\hvz}{\hat{\vec{z}}}
\newcommand{\vz}{\vec{z}}
\newcommand{\vxstar}{\vec{x}^\star}
\newcommand{\vzstar}{\vec{z}^\star}

\newcommand{\vv}{\vec{v}}
\newcommand{\vvstar}{\vec{v}^\star}
\newcommand{\vpistar}{\vpi^\star}

\newcommand{\pr}{\mathbb{P}}

\newcommand{\nakedcite}[1]{\citeauthor{#1}, \citeyear{#1}}

\DeclareMathOperator{\reg}{\mathsf{Reg}}
\DeclareMathOperator{\proj}{\Pi}
\DeclareMathOperator{\inter}{int}

\newcommand{\BF}{B_F}
\newcommand{\BFnoisy}{B_{\noisyF}}

\newcommand{\ind}{r}

\newcommand{\diam}{D}

\NewDocumentCommand{\treeset}{o}{\mathbb{T}\IfNoValueF{#1}{_{#1}}}

\renewcommand{\vec}[1]{\bm{#1}}
\newcommand{\mat}[1]{\mathbf{#1}}

\let\ind\relax

\newcommand{\ind}{r}

\usepackage{hyperref}
\definecolor[named]{ACMPurple}{cmyk}{0.55,1,0,0.15}
\definecolor[named]{ACMDarkBlue}{cmyk}{1,0.58,0,0.21}
\definecolor[named]{lightblue}{RGB}{38,122,196}
\definecolor[named]{lightgreen}{RGB}{28,173,122}
\definecolor[named]{lightred}{RGB}{236,112,99}

\hypersetup{colorlinks,
      linkcolor=ACMPurple,
      citecolor=ACMPurple,
      urlcolor=ACMDarkBlue,
      filecolor=ACMDarkBlue}

\newcommand{\range}[1]{[#1]}

\title{Optimistic Policy Gradient in Multi-Player Markov Games with a Single Controller: \\ Convergence Beyond the Minty Property}

\author[1]{Ioannis Anagnostides\footnote{Part of this work was performed as an intern at Meta AI.}}
\author[2]{Ioannis Panageas}
\author[3]{Gabriele Farina}
\author[4]{Tuomas Sandholm}

\affil[1,4]{Carnegie Mellon University}
\affil[2]{University of California Irvine, Archimedes research unit}
\affil[3]{Massachusetts Institute of Technology}
\affil[4]{Strategy Robot, Inc.}
\affil[4]{Strategic Machine, Inc.}
\affil[4]{Optimized Markets, Inc.}
\affil[ ]{\texttt {ianagnos@cs.cmu.edu}, \texttt{ipanagea@ics.uci.edu}, \texttt{gfarina@mit.edu}, and \texttt{sandholm@cs.cmu.edu}}

\begin{document}

\maketitle

\begin{abstract}
Policy gradient methods enjoy strong practical performance in numerous tasks in reinforcement learning. Their theoretical understanding in multiagent settings, however, remains limited, especially beyond two-player competitive and potential Markov games. In this paper, we develop a new framework to characterize \emph{optimistic} policy gradient methods in multi-player Markov games with a \emph{single controller}. Specifically, under the further assumption that the game exhibits an \emph{equilibrium collapse}, in that the marginals of coarse correlated equilibria (CCE) induce Nash equilibria (NE), we show convergence to \emph{stationary} $\epsilon$-NE in $O(1/\epsilon^2)$ iterations, where $O(\cdot)$ suppresses polynomial factors in the natural parameters of the game. Such an equilibrium collapse is well-known to manifest itself in two-player zero-sum Markov games, but also occurs even in a class of multi-player Markov games with \emph{separable interactions}, as established by recent work. As a result, we bypass known complexity barriers for computing stationary NE when either of our assumptions fails. Our approach relies on a natural generalization of the classical \emph{Minty property} that we introduce, which we anticipate to have further applications beyond Markov games.
\end{abstract}

\section{Introduction}
\label{sec:intro}

Realistic strategic interactions typically occur in stateful multiagent environments in which agents' decisions do not only determine their immediate rewards, but they also shape the next state of the system. Multiagent reinforcement learning (MARL), endowed with game-theoretic principles, furnishes a rigorous framework whereby artificial agents with strong performance guarantees can be developed even in such complex and volatile environments. Indeed, algorithmic advances in MARL have been translated to exciting empirical breakthroughs in grand AI challenges, covering two-player competitive games~\citep{Bowling15:Heads,Brown17:Superhuman,Moravvcik17:DeepStack,Perolat22:Mastering}, as well as popular multi-player games~\citep{Brown19:Superhuman,Meta22:Human}. In spite of those remarkable developments, our theoretical understanding is still lagging behind, especially in multi-player games; this is precisely the primary focus of our paper. 

In particular, we operate in the canonical framework of \emph{Markov} (aka. stochastic) games~\citep{Shapley53:Stochastic,Zhang19:Multi}, which captures multiagent Markov decision processes. Such settings have been the subject of intense scrutiny in recent years, with a flurry of results emerging for computing \emph{Nash equilibria (NE)}---the standard game-theoretic equilibrium concept---in either two-player \emph{zero-sum} games or multi-player \emph{cooperative} games; our synopsis in \Cref{sec:related} features numerous such developments. Algorithmic advances beyond those classes of games are scarce in the literature, and have been considerably impeded by recently established computational barriers for \emph{stationary} NE even in turn-based two-player Markov games~\citep{Daskalakis22:Complexity,Jin23:Complexity}; besides those recent lower bounds, any student of algorithmic game theory should also come to terms with the intrinsic intractability of NE even in one-shot (stateless) general-sum games~\citep{Daskalakis06:Complexity,Chen09:Settling}. Yet, characterizing classes of games that elude those computational barriers is recognized as an important research direction in this line of work.

Our second key motivation---which will naturally coalesce with the considerations described above---is to characterize the behavior of \emph{policy gradient} methods~\citep{Agarwal21:Theory} in Markov games. Such techniques are especially natural from an optimization standpoint, and enjoy strong practical performance in a number of tasks~\citep{Schulman15:Trust,Schulman17:Proximal}. Furthermore, unlike other popular methods, they are amenable to function approximation~\citep{Sutton99:Policy}, thereby enabling to tackle enormous action spaces under continuous parameterizations. 

In light of the inability of traditional gradient-based methods to converge even in normal-form zero-sum games~\citep{Mertikopoulos18:Cycles,Vlatakis20:No}, we focus here on analyzing \emph{optimistic} gradient descent (henceforth $\ogd$). Optimism has been a crucial ingredient in attaining convergence in monotone settings and beyond~\citep{Cai22:Finite,Gorbunov22:Last,Golowich20:Last,Vankov23:Last}, but its role is not well-understood even in two-player zero-sum Markov games. In this paper, we take an important step towards closing this gap, which will uncover as a byproduct a new class of multi-player Markov games for which we can compute efficiently stationary Nash equilibria.

\subsection{Our results}

To contextualize our approach, we first have to highlight a classical condition in variational inequalities (VIs) which guarantees convergence under certain first-order methods; namely, the so-called \emph{Minty property}~\citep{Facchinei03:Finite,Mertikopoulos19:Optimistic}. A great number of existing results in optimization---not least in the multiagent setting---leverage that condition to analyze the behavior of learning algorithms. Unfortunately, \citet{Daskalakis20:Independent} observed that the Minty property fails even in simple two-player Markov games with a single controller (recalled in \Cref{prop:Minty-fails}). Furthermore, although several relaxations of the Minty property have been proposed (see our overview in \Cref{sec:related}), none has been able to capture such settings, thereby leaving open whether optimistic policy gradient methods converge.

In this context, our first main contribution is to introduce a generalization of the Minty property (\Cref{property:genMVI}) which addresses the aforementioned difficulties by capturing a broad class of multi-player Markov games. Specifically, our condition is more permissive in two crucial aspects. First, it allows distorting the underlying operator by a certain well-behaved function; as we explain in \Cref{sec:beyMinty}, this modification already suffices to subsume the counterexample of~\citet{Daskalakis20:Independent}---and generalizations thereof. The second modification relaxes the pointwise aspect of the original Minty property into an average guarantee, in the precise sense of \Cref{prop:tm-MVI}. 

Now the upshot is that $\ogd$---under a suitable parameterization---still converges to an $\epsilon$-strong solution of the induced VI problem after $T = O_\epsilon(1/\epsilon^2)$ iterations even under our more permissive criterion (\Cref{theorem:genMVI}), where the notation $O_\epsilon(\cdot)$ here suppresses polynomial factors in all natural parameters of the problem. We further establish that this guarantee is robust in the presence of perturbations akin to \emph{relative deterministic noise} (\Cref{cor:noisy})---a ubiquitous model in control theory and optimization---and a certain slackness in our condition (\Cref{cor:gammagenMVI}); the latter extension turns out to be crucial to capture policy optimization under greedy exploration.

As we have alluded to, the main application of our general theory targets multi-player Markov games, formally introduced in \Cref{sec:prels}. In light of the inherent computational barriers described earlier, we need to impose additional structure to obtain meaningful guarantees. Our first assumption is that the underlying Markov game exhibits a certain \emph{equilibrium collapse}, in that the marginals of \emph{coarse correlated equilibria} induce Nash equilibria (\Cref{def:equilcol}). It is well-known that such is the case in two-player zero-sum games, but recent work~\citep{Kalogiannis23:Zero,Park23:Multi} has also revealed that equilibrium collapse persists even in a class of multi-player zero-sum games with \emph{separable interactions}---building on a similar result in normal-form \emph{polymatrix} games~\citep{Cai16:Zero}. Yet, perhaps surprisingly and in stark contrast to normal-form games, equilibrium collapse alone does not suffice to enable efficient computation of \emph{stationary} Nash equilibria~\citep{Daskalakis22:Complexity,Jin23:Complexity}. For this reason, we further posit that the game admits a \emph{single} controller, a quite classical setting as surveyed in \Cref{sec:related}. The upshot now is that under those two assumptions, our condition that generalizes the Minty property holds (\Cref{lemma:property}), which brings us to one of our main results.

\begin{theorem}[Informal; precise version in \Cref{theorem:main}]
    Consider any multi-player Markov game $\cG$ with a single controller. If $\cG$ exhibits equilibrium collapse, there is a $\poly(|\cG|, 1/\epsilon)$ algorithm that receives gradient feedback and computes a stationary $\epsilon$-Nash equilibrium.
\end{theorem}

Above, we denote by $\poly(|\cG|)$ a polynomial in the natural parameters of the game; the precise version appears as \Cref{theorem:main}. In light of existing hardness results for computing stationary NE even in turn-based two-player Markov games~\citep{Daskalakis22:Complexity,Jin23:Complexity}, it is unlikely that the assumption of having a single controller can be significantly broadened. We also consider our theory investigating tractability beyond the Minty property to have interest beyond Markov games, but this is left for future work.
\section{Preliminaries on Markov Games}
\label{sec:prels}

In this section, we provide the necessary preliminaries on Markov games. For further background on Markov decision processes (MDPs), we refer to the works of~\citet{Sutton18:Reinforcement,Szepesvari22:Algorithms,Bucsoniu10:Multi}.

\paragraph{Notation} We let $\N = \{1, 2, \dots, \}$ denote the set of natural numbers and $\N^* \defeq \N \cup \{0\}$. For $n \in \N$, we use the shorthand notations $\range{n} \defeq \{1, \dots, n\}$ and $\range{n}^* \defeq \{0, 1, \dots, n\}$. For a vector $\vz \in \R^d$, we often use the variable $\ind \in \range{d}$ to index its coordinates, so that the $\ind$th coordinate is accessed by $\vz[\ind]$. The inequality $\vz \leq \cdot$ is to be interpreted coordinate-wise. For two vectors $\vz, \vz' \in \R^d$, we denote by $\vz \circ \vz' \in \R^d$ their Hadamard product: $(\vz \circ \vz') [\ind] \defeq \vz[\ind] \cdot \vz'[\ind]$, for all $\ind \in \range{d}$.

Moreover, we will let $\cX$ represent a convex nonempty and compact subset of a Euclidean space. We denote by $\diam_{\cX}$ its $\ell_2$ diameter. A function $F : \cX \to \cX$ is called $L$-Lipschitz continuous (with respect to the $\ell_2$ norm $\|\cdot\|_2$) if $\| F(\vx) - F(\vx') \|_2 \leq L \|\vx - \vx'\|_2$, for any $\vx, \vx' \in \cX$; a differentiable function is called $L$-smooth if its gradient is $L$-Lipschitz continuous. Finally, to lighten the exposition, we will often use the $O_n(\cdot)$ notation to indicate the dependency of a function solely on parameter $n$.

\paragraph{Markov games} We are interested in analyzing the convergence of policy gradient methods in multi-player Markov (aka. stochastic) games~\citep{Shapley53:Stochastic} in the tabular regime. In such games, each player repeatedly elects actions within a multiagent MDP so as to maximize a reward function. Formally, a multi-player Markov game $\cG$ is specified by a tuple $(\cN, \cS, \{ \cA_i \}_{i=1}^n, \pr, \{ R _i \}_{i=1}^n, \zeta, \vrho) \eqqcolon \cG$, whose constituents are defined as follows.

\begin{itemize}[noitemsep] 
    \item $\cN \defeq \range{n}$ is the set of players (or agents); 
    \item $\cS$ is a finite \emph{state space};
    \item $\cA_i$ is the finite and nonempty set of available actions for each player $i \in \range{n}$ (for simplicity, and without loosing any generality, we posit that the action set does not depend on the underlying state); further, the joint action set is denoted by $\cA \defeq \bigtimes_{i=1}^n \cA_i$;
    \item $\pr$ is the transition probability function, so that $\pr(s' | s, \vec{a})$ represents the probability of transitioning to state $s' \in \cS$ starting from state $s \in \cS$ under the joint action $\vec{a} \in \cA$; 
    \item $R_i : \cS \times \cA \to [-1, 1] $ is the (normalized) reward function of player $i \in \range{n}$, so that $R_i(s, \vec{a})$ represents the instantaneous reward when players select $\vec{a} \in \cA$ in state $s \in \cS$; (For simplicity, the rewards are deterministic.)
    \item $\zeta \defeq \min_{(s, \va) \in \cS \times \cA} (1 - \sum_{s' \in \cS} \pr(s' | s, \va) ) > 0$ is a lower bound on the probability that the game will terminate at some step of the shared MDP; and
    \item $\vrho \in \Delta(\cS)$ is the initial distribution over states, assumed to have full support.
\end{itemize}

\paragraph{Learning algorithms} Learning in such multiagent settings proceeds as follows. At every step $h \in \N^*$ each player $i \in \range{n}$ 1) observes the underlying state $s_h \in \cS$; 2) selects an action $a_{i, h} \in \cA_i$; and 3) subsequently receives some feedback from the environment, to be specified in the sequel. This process is repeated until the game terminates, which indeed occurs with probability $1$ since we assume that $\zeta > 0$; the last step before the game terminates will be denoted by $H \in \N^*$, which is a random variable. 

\paragraph{Policies} A (potentially randomized) \emph{stationary policy} for player $i \in \range{n}$ is a mapping $\vpi_i : \cS \to \Delta(\cA_i)$; that is, a stationary policy remains invariant for all steps $h \in \N^*$. We only consider \emph{Markovian} policies throughout this paper, without explicitly mentioning so. We will assume that players follow direct parameterization so that $\vpi_i \mapsto \vx_i \in \Delta(\cA_i)^{\cS} \eqqcolon \cX_i$ with the strategy $\vx_{i, s}[a_i] \defeq \vpi_i(a_i | s)$ for all $(a_i, s) \in \cA_i \times \cS$. As such, strategies and policies will be used interchangeably. The set of all possible (stationary) policies for player $i \in \range{n}$ will be denoted by $\Pi_i$, while $\Pi \defeq \bigtimes_{i=1}^n \Pi_i$. We will also let $\cX \defeq \bigtimes_{i=1}^n \cX_i$.

\paragraph{Value} The \emph{value function} $V_{i}^{\vpi}(s)$ with respect to an initial state $s \in \cS$ gives the expected reward for player $i \in \range{n}$ under the joint policy $\vpi \defeq (\vpi_1, \dots, \vpi_n) \in \Pi$:
\begin{equation}
    \label{eq:Value}
    V_{i}^{\vpi}(s) \defeq \E_{\vpi} \left[ \sum_{h=0}^H R_i(s_h, \va_h) | s_0 = s \right],
\end{equation}
where the expectation above is taken with respect to the trajectory induced by the joint policy $\vpi \in \Pi$. We also generalize~\eqref{eq:Value} by defining $V_{i}^{\vpi}(\vrho) \defeq \E_{s \sim \vrho} [V^{\vpi}_{i}(s)]$, where we recall that $\vrho \in \Delta(\cS)$. Similarly, the $Q$ function with respect to player $i$ is defined as
\begin{equation*}
    Q_i^{\vpi}: (s, \va) \mapsto \E_{\vpi} \left[ \sum_{h=0}^H R_i(s_h, \va_h) | s_0 = s, \va_{0} = \va \right],
\end{equation*}
where the expectation is again taken over the trajectory induced by $\vpi \in \Pi$. 
In this context, we will assume that each player receives as feedback from the environment the gradient of its value function with respect to its strategy.

\paragraph{Nash equilibrium} Consider any player $i \in \range{n}$, and let $\vmu_{-i} : \cS \to \Delta(\cA_{-i})$ be a potentially correlated policy. We denote a stationary \emph{best response policy} of $i$ under $\vmu_{-i}$ by $\vpi^\dagger_i = \vpi^\dagger_i(\vmu_{-i}) \in \Pi_i$, so that $V^{\dagger, \vmu_{-i}}_{i}(\vrho) \defeq V_i^{\vpi^\dagger_i, \vmu_{-i}}(\vrho)$.\footnote{It is well-known that there is always a stationary policy among the set of best response policies~\citep{Sutton18:Reinforcement}.
}

\begin{definition}
    \label{def:NE}
    A (stationary) \emph{product} policy $\vpistar \in \Pi$ is an $\epsilon$-Nash equilibrium if 
    $$\max_{1 \leq i \leq n} \left\{ V_i^{\dagger, \vpistar_{-i}}(\vrho) - V_i^{\vpistar}(\vrho) \right\} \leq \epsilon$$.
\end{definition}

Finally, for $\vpi \in \Pi$, we define the \emph{state visitation distribution} $d_{s_0}^{\vpi} \in \Delta(\cS)$ by $d_{s_0}^{\vpi}[s] \propto \sum_{h \in \N^*} \pr^{\vpi}(s_h = s | s_0)$, and $d_{\vrho}^{\vpi} \defeq \E_{s_0 \sim \vrho} [ d_{s_0}^{\vpi} ]$. It will also be useful to consider the unnormalized counterparts of those distributions: $\tilde{d}_{s_0}^{\vpi}[s] = \sum_{h \in \N^*} \pr^{\vpi}(s_h = s | s_0)$ and $\tilde{d}_{\vrho}^{\vpi} \defeq \E_{s_0 \sim \vrho} [\tilde{d}_{s_0}^{\vpi}]$.

\section{Convergence Beyond the Minty Property}
\label{sec:beyMinty}

A classical condition that guarantees tractability for a variational inequality (VI) problem is the so-called \emph{Minty property}~\citep{Facchinei03:Finite}. To be precise, let $F : \cX \to \cX$ be a single-valued operator. The Minty property postulates the existence of a point $\vxstar \in \cX$ such that
\begin{equation}
    \label{eq:Minty}
    \langle \vx - \vxstar, F(\vx) \rangle \geq 0, \quad \forall \vx \in \cX.
\end{equation}

By now, there has been significant progress on understanding convergence of first-order methods under the Minty property. Unfortunately, and crucially for the purpose of this work, even two-player zero-sum Markov games fail to satisfy~\eqref{eq:Minty}, as was first observed by~\citet{Daskalakis20:Independent}. In particular, they studied a simple two-player zero-sum Markov game known as Von Neumann's \emph{ratio game}~\citep{Neumann45:Model}, given by
\begin{equation}
    \label{eq:ratio}
    V(\vx_1, \vx_2) \defeq \frac{\vx_1^\top \mat{R} \vx_2 }{\vx_1^\top \mat{S} \vx_2},
\end{equation}
where $\vx_1 \in \Delta(\cA_1) \eqqcolon \cX_1, \vx_2 \in \Delta(\cA_2) \eqqcolon \cX_2$, and $\mat{R}, \mat{S} \in \R^{\cA_1 \times \cA_2}$. It is  further assumed that $\vx_1^\top \mat{S} \vx_2 \geq \zeta$, for some parameter $\zeta > 0$. The following proposition underlies much of the difficulty of analyzing policy gradient methods even under the simple ratio game~\eqref{eq:ratio}.

\begin{proposition}[\nakedcite{Daskalakis20:Independent}]
    \label{prop:Minty-fails}
    Fix any scalars $\epsilon, s \in (0,1)$, and suppose that
    \begin{equation}
        \label{eq:R-S}
        \mat{R} \defeq
        \begin{pmatrix}
            -1 & \epsilon \\
            - \epsilon & 0
        \end{pmatrix} 
        \quad \textrm{and} \quad
        \mat{S} \defeq
        \begin{pmatrix}
            s & s \\
            1 & 1
        \end{pmatrix}.
    \end{equation}
    Then, the ratio game induced by the matrices in~\eqref{eq:R-S} fails to satisfy the Minty property~\eqref{eq:Minty}.
\end{proposition}

Notwithstanding the above realization, empirical simulations suggest that optimistic policy gradient methods do in fact exhibit convergent behavior. As a result, a criterion more robust than the Minty property is needed. This is precisely the primary subject of this section. 

Before we proceed with our generalized condition, let us make a further observation regarding the ratio game defined in~\Cref{prop:Minty-fails} that will be useful in the sequel: that game admits a single controller---the transition probabilities depend solely on the strategy of one of the players; indeed, we note that $\vx_1^\top \mat{S} \vx_2 = \vx_1^\top \vec{s}$ for any $(\vx_1, \vx_2) \in \cX_1 \times \cX_2$, where $\vec{s} = (s, 1)$---and thereby does not depend on $\vx_2$.

Now, to address the aforementioned difficulties, we introduce and study a new condition, described below. 

\begin{property}[Generalized Minty property]
    \label{property:genMVI}
    Let $F : \cX \to \cX$ be such that $\cX = \bigtimes_{\ind=1}^d \cZ_r$ for $d \in \N$, and $\vec{1}_{\cZ_{\ind}}$ be the vector with $1$ for all entries corresponding to the component $\cZ_{\ind}$, and $0$ otherwise. Suppose further that $A: \cX \to \cX$ and $W : \cX \to \cX$ are functions such that
    \begin{itemize}
        \item $A(\vx) \defeq \sum_{\ind=1}^d a_\ind(\vx) \vec{1}_{\cZ_\ind}$, where each $a_\ind : \cX \to \R$ is $\alpha$-Lipschitz continuous; $0 < \ell \leq A(\vx) \leq h$; and
        \item $W(\vx) \defeq \sum_{\ind=1}^d w_{\ind}(\vx) \vec{1}_{\cZ_{\ind}}$; $0 < \ell \leq W(\vx) \leq h$.
    \end{itemize}
    We say that the induced VI problem satisfies the $(\alpha, \ell, h)$-generalized Minty property 
    if there exists $\vxstar \in \cX$ so that
    \begin{equation}
        \label{eq:genMVI}
        \langle \vx - \vxstar, F(\vx) \circ A(\vx) \circ W(\vxstar) \rangle \geq 0, \quad \forall \vx \in \cX,
    \end{equation}
    where $\circ$ denotes component-wise multiplication.
\end{property}

Several remarks are in order regarding this property. First, a key assumption is that the underlying joint strategy space $\cX$ can be decomposed as a Cartesian product, and that the functions $A$ and $W$ adhere to that structure. It is evident that \Cref{property:genMVI} is more general than~\eqref{eq:Minty} since one can simply take $A$ and $W$ to be constant functions. In fact, when $d = 1$ the two conditions are equivalent; it is precisely the product structure of $\cX$---which is inherently present in multi-player games---that makes \Cref{property:genMVI} interesting. It is also worth noting a related condition appearing in~\citep[Appendix C.5]{Harris23:Meta}, although it did not have any algorithmic implications.

Let us now relate \Cref{property:genMVI} to the difficulty exposed by~\Cref{prop:Minty-fails} in the context of the ratio game. One can show that if $\vxstar \in \cX_1 \times \cX_2$ is a Nash equilibrium of the ratio game, then if we take $A(\vx_1, \vx_2)$ and $W(\vx_1, \vx_2)$ as
$$ (\vx_1^\top \vec{s} \overbrace{(1, \dots, 1)}^{|\cA_1|}, \overbrace{(1, \dots, 1)}^{|\cA_2|}), \Big(\frac{1}{\vx_1^\top \vec{s}} \overbrace{(1, \dots, 1)}^{|\cA_1|}, \overbrace{(1, \dots, 1)}^{|\cA_2|} \Big),$$ respectively, then~\eqref{eq:genMVI} is satisfied (in this particular application, $d = 2$). Furthermore, having assumed that $\vx_1^\top \mat{S} \vx_2 \geq \zeta > 0$, we also have control over the lower bound $\ell$ (as well as the upper bound $h$); naturally, taking $\ell$ arbitrarily small trivializes \Cref{property:genMVI}, and so the interesting regime occurs when $\ell$ is bounded away from $0$---this also becomes evident from the guarantee of \Cref{theorem:genMVI}. This observation regarding the VI induced by the ratio game is in fact non-trivial, and it is a byproduct of the minimax theorem shown by~\citet{Shapley53:Stochastic}; in \Cref{sec:OGD-Markov}, we will prove this property in much greater generality.

As we shall see, \Cref{property:genMVI} is already permissive enough to lead beyond known results. Nevertheless, to obtain as general results as possible, we next introduce a further extension of \Cref{property:genMVI}.

\begin{property}[Average version of \Cref{property:genMVI}]
    \label{prop:tm-MVI}
    Under the preconditions of \Cref{property:genMVI} with respect to some triple $(\alpha, \ell, h) \in \R^3_{> 0}$, we say that the induced VI problem satisfies the \emph{average} $(\alpha, \ell, h)$-generalized Minty property if for any sequence $\sigma^{(T)} \defeq (\vx^{(t)})_{1 \leq t \leq T}$ there exists $\cX \ni \vxstar = \vxstar(\sigma^{(T)})$ so that
    \begin{equation}
        \label{eq:tm-MVI}
        \sum_{t=1}^T \langle \vx^{(t)} - \vxstar, F(\vx^{(t)}) \circ A(\vx^{(t)}) \circ W(\vxstar) \rangle \geq 0.
    \end{equation}
\end{property}

\Cref{property:genMVI} clearly implies \Cref{prop:tm-MVI} as a suitable $\vxstar \in \cX$ would make every term in the summand~\eqref{eq:tm-MVI} nonnegative; we have found that the additional generality of the latter property is useful for some applications (see \Cref{sec:OGD-Markov}).

We are now ready to proceed to the main result of this section, which concerns the behavior of the update rule
\begin{equation}
    \tag{$\mathtt{OGD}$}
    \label{eq:OGD}
    \begin{split}
    \vx^{(t)} \defeq \proj_\cX(\hvx^{(t)} - \eta A(\vx^{(t-1)}) \circ F(\vx^{(t-1)})),\\
    \hvx^{(t+1)} \defeq \proj_\cX(\hvx^{(t)} - \eta A(\vx^{(t)}) \circ F(\vx^{(t)})),
    \end{split}
\end{equation}
for $t \in \N$. Above, $\eta > 0$ is the learning rate; $\proj_{\cX}(\cdot)$ is the Euclidean projection operator; and $\vx^{(0)} = \hvx^{(1)} \in \cX$ is an arbitrary initialization. The update rule~\eqref{eq:OGD} is the familiar optimistic gradient descent method~\citep{Chiang12:Online,Rakhlin13:Online}, but with an important twist: the operator $F(\vx^{(t)})$ is now replaced by $ A(\vx^{(t)}) \circ F(\vx^{(t)})$, where $A : \cX \to \cX$ is a problem-specific function---in direct correspondence with \Cref{property:genMVI}; this can be simply viewed as incorporating a time-varying but non-vanishing learning rate. We remark that it is assumed that $A$ can be accessed in order to perform the update rule~\eqref{eq:OGD}; this assumption will be discussed and addressed in the context of our applications in \Cref{sec:OGD-Markov}.
Below, we show that \Cref{prop:tm-MVI} is indeed sufficient to guarantee tractability for the induced VI problem, in the following formal sense.

\begin{restatable}{theorem}{genMVI}
    \label{theorem:genMVI}
    Let $\cX = \bigtimes_{\ind=1}^d \cZ_\ind$ for some $d \in \N$ and $F : \cX \to \cX$ be an $L$-Lipschitz continuous operator with $\BF \defeq \max_{1 \leq \ind \leq d} \| F_\ind \|_2$. Suppose further that the average $(\alpha, \ell, h)$-generalized Minty property (\Cref{prop:tm-MVI}) holds. Then, for any $\epsilon > 0$, after $T \geq \frac{2 \diam_{\cX}^2 h}{\ell \epsilon^2}$ iterations of~\eqref{eq:OGD} with learning rate $\eta \leq \frac{1}{4} \sqrt{\frac{\ell}{h^3 L^2 + h \BF^2 \alpha^2 d}}$ there is a point $\vx^{(t)} \in \cX$ such that for any $\vxstar \in \cX$,
    \begin{equation*}
        \langle \vx^{(t)}, F(\vx^{(t)}) \rangle - \langle \vxstar, F(\vx^{(t)}) \rangle \leq 2 d \left( \frac{\max_{1 \leq r \leq d} \diam_{\cZ_r}}{\eta \ell} + \frac{h \BF}{\ell} \right) \epsilon.
    \end{equation*}
\end{restatable}

\paragraph{Proof sketch} The proof of this theorem is deferred to \Cref{appendix:beyMVI}, but we briefly describe the key ingredients here. In a nutshell, we analyze optimistic gradient descent~\eqref{eq:OGD} following the regret analysis of optimistic mirror descent (\Cref{prop:RVU}) in the context of multi-player games~\citep{Rakhlin13:Optimization,Syrgkanis15:Fast}; more precisely, we essentially view each component over $\cZ_{\ind}$, comprising the Cartesian product $\cX \defeq \bigtimes_{\ind=1}^d \cZ_{\ind}$, as a separate player. The twist is that---in accordance with~\eqref{eq:OGD}---the observed utility is taken to be $F_\ind(\vx^{(t)}) \circ A_\ind(\vx^{(t)})$, instead of $F_\ind(\vx^{(t)})$, where $F_\ind$ is the $\ind$th component of $F$. Importantly, the structure imposed on $A(\vx)$ by~\Cref{prop:tm-MVI} enables us to show that a suitable \emph{weighted} notion of regret enjoys a certain upper bound independent of both $A$ and $W$. Thus, leveraging~\eqref{eq:tm-MVI}, we are able to show---following earlier work~\citep{Anagnostides22:On,Zhang22:No}---that the second-order path lengths of the dynamics are bounded (\Cref{cor:pathlength}). Then, \Cref{theorem:genMVI} follows by the assumption that $0 < \ell \leq A(\vx) \leq h$; that is, incorporating $A(\vx)$ into the update rule \eqref{eq:OGD} does not distort by much the underlying operator $F$.

A point $\vx^{(t)}$ such that $\langle \vx^{(t)} - \vxstar, F(\vx^{(t)}) \rangle \leq \epsilon$ for any $\vxstar \in \cX$---as in the guarantee of \Cref{theorem:genMVI}---is known as an $\epsilon$-approximate solution to the \emph{Stampacchia} VI problem (aka. an $\epsilon$-approximate strong solution). To make this guarantee more concrete, and to connect it with the forthcoming application in \Cref{sec:OGD-Markov}, let us consider an $n$-player game so that $F = (F_1, \dots, F_n)$ and $F_i \defeq -\nabla_{\vx_i} u_i(\vx)$, where $u_i : \cX \to \R$ is the differentiable utility of player $i \in \range{n}$.

\begin{corollary}
    \label{cor:local}
    Under the preconditions of \Cref{theorem:genMVI}, we can compute a point $\vx \in \cX$ after a sufficiently large $T = O_\epsilon(1/\epsilon^2)$ iterations of~\eqref{eq:OGD}, for any $\epsilon > 0$, such that
    \begin{enumerate}
        \item if each $u_i(\vx_i, \cdot)$ is $L$-smooth, then for any player $i \in \range{n}$ and $\vxstar_i \in \cX_i$ with $\| \vxstar_i - \vx_i \|_2 \leq \delta$, $u_i(\vx) - u_i(\vxstar_i, \vx_{-i}) \geq - \epsilon - \frac{L}{2} \delta^2;$
        \label{item:local}
        \item if each $u_i(\vx_i, \cdot)$ is gradient dominant, then for any player $i \in \range{n}$ and $\vxstar_i \in \cX_i$, $u_i(\vx) - u_i(\vxstar_i, \vx_{-i}) \geq - \epsilon.$
        \label{item:GD}
    \end{enumerate}
\end{corollary}
To be precise, the (per-player) gradient dominance property postulates that 
$$u_i(\vx) - \max_{\vxstar_i \in \cX_i} u_i(\vxstar_i, \vx_{-i}) \geq G \min_{\vxstar_i \in \cX_i} \langle \vx_i - \vxstar_i, \nabla_{\vx_i} u_i(\vx) \rangle$$ for all $\vx \in \cX$, where $G > 0$ is some parameter. As such, \Cref{item:GD} follows directly by definition and \Cref{theorem:genMVI}. \Cref{item:local} above is more permissive, but only yields a local optimality guarantee. Still, it turns out that computing such points is hard even in smooth min-max optimization~\citep{Daskalakis21:Complexity,Daskalakis22:Stay}; more precisely, \Cref{item:local} is interesting in the local regime $\delta < \sqrt{\frac{2\epsilon}{L}}$; see~\citep[Definition 1.1]{Daskalakis21:Complexity}; other notions of local optimality have also been studied in the literature~\citep{Jin20:What}, but this is not in our scope here.

Before we conclude this section, let us highlight some interesting extensions of \Cref{theorem:genMVI}. First, one can further broaden the scope of \Cref{prop:tm-MVI} by replacing the right-hand side of~\eqref{eq:tm-MVI} by $-\gamma T$, for some parameter $\gamma \in \R_{\geq 0}$. In \Cref{cor:gammagenMVI}, we show that we can then compute an $O_{\epsilon, \gamma}(\sqrt{\gamma} + \epsilon)$-approximate strong solution. This particular relaxation turns out to be crucial to capture policy parameterization under $\Theta_\gamma(\gamma)$-greedy exploration (\Cref{remark:exploration} elaborates on this point). In such settings, one has control over the parameter $\gamma$, and so by taking $\gamma \defeq \epsilon^2$ we can generalize the guarantee of \Cref{theorem:genMVI}.

Our second extension concerns the behavior of~\eqref{eq:OGD} in the presence of noise. Our model of perturbation is akin to the standard relative deterministic noise, wherein the error is proportional to the distance from optimality, for an appropriate notion of distance~\citep{Polyak87:Introduction,Lessard16:Analysis}. More precisely, for parameters $\rho, \delta > 0$, we assume access to a noisy operator $\noisyF : \cX \to \cX$ such that $\| \noisyF(\vx) - F(\vx) \|_2 \leq \delta \cdot \eqgap(\vx)$, where $\eqgap(\vx) : \cX \ni \vx \mapsto \max_{\vxstar \in \cX} \langle \vx - \vxstar, F(\vx) \rangle $ represents the equilibrium gap. We further posit that $\noisyF$ satisfies a relaxed version of \Cref{prop:tm-MVI} in which the right-hand side of~\eqref{eq:tm-MVI} can be as small as $-\rho \sum_{t=1}^T (\eqgap(\vx^{(t)}))^2$. In this context, \Cref{cor:noisy} reassures us that the conclusion of \Cref{theorem:genMVI} is robust if $\delta$ and $\rho$ are small enough.
\section{Optimistic Policy Gradient in Multi-Player Markov Games}
\label{sec:OGD-Markov}

In this section, we leverage the theory developed earlier in \Cref{sec:beyMinty} in order to characterize optimistic policy gradient methods in multi-player Markov games. In light of the inherent hardness of computing Nash equilibria in general-sum games, we will restrict our attention to more structured classes of Markov games. The first assumption we consider can be viewed as a natural counterpart of the Minty property, but with respect to the value functions---without linearizing by taking the gradients.

\begin{assumption}
    \label{assumption:minimax}
    Let $\cG$ be a Markov game. There exists a joint policy $(\vpistar_1, \dots, \vpistar_n) \in \Pi$ such that
    \begin{equation*}
        \sum_{i=1}^n V_{i}^{\vpistar_i, \vpi_{-i}}(\vrho) - \sum_{i=1}^n V_{i}^{\vpi}(\vrho) \geq 0, \forall (\vpi_1, \dots, \vpi_n) \in \Pi.
    \end{equation*}
\end{assumption}

Crucially, unlike the Minty property~\eqref{eq:Minty}, \Cref{assumption:minimax} subsumes two-player zero-sum (Markov) games. Indeed, \citet{Shapley53:Stochastic} proved that there exist policies $(\vpistar_1, \vpistar_2) \in \Pi$ such that
\begin{equation*}
    V^{\vpistar_1, \vpi_2}(\vrho) \leq V^{\vpistar_1, \vpistar_2}(\vrho) \leq V^{\vpi_1, \vpistar_2}(\vrho), \quad \forall (\vpi_1, \vpi_2) \in \Pi.
\end{equation*}
Here, $V_{1}(\vrho) \defeq - V({\vrho})$ and $V_{2}(\vrho) \defeq V(\vrho)$ (since the game is zero-sum). The above display establishes \Cref{assumption:minimax} since $ V_{1}^{\vpistar_1, \vpi_2}(\vrho) + V_{2}^{\vpi_1, \vpistar_2}(\vrho) \geq 0$. In other words, \Cref{assumption:minimax} is a byproduct of Shapley's minimax theorem. 

It is worth noting that any (stationary) Nash equilibrium $(\vpistar_1, \dots, \vpistar_n) \in \Pi$ satisfies
\begin{equation*}
    \sum_{i=1}^n V_{i}^{\vpistar}(\vrho) - \sum_{i=1}^n V_{i}^{\vpi_i, \vpistar_{-i}}(\vrho) \geq 0, \forall (\vpi_1, \dots, \vpi_n) \in \Pi,
\end{equation*}
which closely resembles the condition of~\Cref{assumption:minimax}. However, unlike \Cref{assumption:minimax}, the above condition always holds since (stationary) NE always exist~\citep{Fink64:Equilibrium}.

As it will become clear, \Cref{assumption:minimax} is naturally associated with \Cref{property:genMVI}. We also introduce a more permissive assumption in direct correspondence with \Cref{prop:tm-MVI}.

\begin{assumption}
    \label{ass:nonnegativeregret}
    Let $\cG$ be a Markov game. For any sequence of product policies $\sigma^{(T)} \defeq (\vpi^{(t)})_{1 \leq t \leq T}$, there exists $\Pi \ni \vpistar = \vpistar(\sigma^{(T)})$ such that
    \begin{equation*}
        \sum_{t=1}^T \sum_{i=1}^n V_{i}^{\vpistar_i, \vpi^{(t)}_{-i}}(\vrho) - \sum_{t=1}^T \sum_{i=1}^n V_{i}^{\vpi^{(t)}}(\vrho) \geq 0.
    \end{equation*}
\end{assumption}

Beyond the two-player zero-sum setting, we first show that \Cref{ass:nonnegativeregret} is satisfied for the class of \emph{zero-sum polymatrix Markov} games~\citep{Kalogiannis23:Zero} (see also~\citep{Park23:Multi}). 

\paragraph{Polymatrix zero-sum Markov games} A polymatrix game is based on an undirected graph $G = (V, E)$. Each node $i \in V$ is (uniquely) associated with a player, while every edge $\{i, i'\} \in E$ represents a pairwise interaction between players $i$ and $i'$. It is assumed that the reward of each player is given by the sum of the rewards from each game engaged with its neighbors. The zero-sum aspect imposes that the sum of the players' rewards is $0$. Such games were investigated by~\citet{Cai16:Zero} (see also~\citet{Even-Dar09:Convergence} for a more general treatment) under the normal form representation. For the Markov setting, \citet{Kalogiannis23:Zero} further assumed that in each state there is a single player (not necessarily the same) whose actions determine the transition probabilities to the next state. For that class of games, with a careful examination of their analysis we are able to show the following result.

\begin{restatable}{proposition}{polymatrixmvi}
    \label{prop:polymatrix-mvi}
    \Cref{ass:nonnegativeregret} is satisfied for any polymatrix zero-sum Markov game.
\end{restatable}

In fact, this result is a byproduct of a more general characterization that we prove. To explain our result, we first recall the concept of a \emph{coarse correlated equilibrium (CCE)}, which relaxes \Cref{def:NE} by allowing \emph{correlated} policies. We will further use the concept of an $\epsilon$-\emph{average} CCE (henceforth $\epsilon$-ACCE), a relaxation of CCE in which the sum---instead of the maximum as in CCE---of the players' deviation benefits must be at most $\epsilon$. Precise definitions are deferred to \Cref{sec:furtherprels}, as they are not crucial for the purpose of this section. In this context, we introduce the following definition.

\begin{definition}[Equilibrium collapse]
    \label{def:equilcol}
    Let $\cG$ be a Markov game. We say that $\cG$ exhibits \emph{equilibrium collapse} if there is a $C = C(\cG) \in \R_{> 0}$ such that for any stationary $\epsilon$-ACCE $\vmu \in \Delta(\cA)^{\cS}$ of $\cG$, the marginal policies $(\vpi_1, \dots, \vpi_n) = (\vpi_1(\vmu), \dots, \vpi_n(\vmu))$ form a $(C \epsilon)$-Nash equilibrium of $\cG$.
\end{definition}

We remark that the prior work on zero-sum polymatrix Markov games established equilibrium collapse with respect to $\epsilon$-CCE~\citep{Kalogiannis23:Efficiently}, but their argument readily carries over for ACCE as well. \Cref{prop:polymatrix-mvi} is thus implied by the following result.

\begin{restatable}{proposition}{nonneg}
    \label{prop:acce}
    \Cref{ass:nonnegativeregret} is satisfied in any Markov game $\cG$ exhibiting equilibrium collapse per \Cref{def:equilcol}.
\end{restatable}

Having justified \Cref{assumption:minimax,ass:nonnegativeregret}, we now proceed to establishing \Cref{property:genMVI}. Taking a step back, one might hope that equilibrium collapse (in the sense of \Cref{def:equilcol}) would already suffice to efficiently compute stationary NE---as in the case of normal-form games. However, recent lower bounds~\citep{Daskalakis22:Complexity,Jin23:Complexity} dispel any such hopes, thereby necessitating additional structure in order to elude those intractability barriers. This is precisely where the admission of a single controller comes into play, an assumption crucial for establishing \Cref{prop:tm-MVI}. Indeed, this is shown in the following key lemma, which relies on the expression of the difference of the value function (\Cref{lemma:Valdiff}) and the connection between the $Q$ function and the gradient of the value function (\Cref{lemma:Q-V}). In accordance with our theory in \Cref{sec:beyMinty}, we let $F_\cG(\vx) \defeq - (\nabla_{\vx_1} V_{1}(\vrho), \dots, \nabla_{\vx_n} V_{n}(\vrho))$; see \Cref{remark:diff} for a clarification regarding differentiability of the value function in spite of the empty interior of $\cX$.

\begin{restatable}{lemma}{prop}
    \label{lemma:property}
    Consider a Markov game $\cG$, and let $\Lambda_i(\vx, \vxstar)[s, a_i] \defeq \frac{\Tilde{d}_{\vrho}^{\vpistar_i, \vpi_{-i}}[s]}{\Tilde{d}_{\vrho}^{\vpi}[s]}$ for $i \in \range{n}$ and $(s, a_i) \in \cS \times \cA_i$. Further, let $\Lambda(\vx, \vxstar) \defeq (\Lambda_1(\vx, \vxstar), \dots, \Lambda_n(\vx, \vxstar))$. If \Cref{assumption:minimax} holds, then there exists $\vxstar \in \cX$ such that
    \begin{equation}
        \label{eq:gengenMVI}
        \langle \vx - \vxstar, F(\vx) \circ \Lambda(\vx, \vxstar) \rangle \geq 0, \quad \forall \vx \in \cX.
    \end{equation}
    In particular, if $\cG$ admits a single controller, denoted by $\cntrl$, then \Cref{property:genMVI} holds with
    \begin{equation*}
        A_i(\vx)[s, a_i] \defeq \begin{cases}
            1 &: \textrm{if } i \neq \cntrl\\
            \left( \Tilde{d}_{\vrho}^{\vpi_i}[s] \right)^{-1} &: \textrm{if } i = \cntrl,
        \end{cases}
    \end{equation*}
    and
    \begin{equation*}
        W_i(\vxstar)[s, a_i] \defeq \begin{cases}
            1 &: \textrm{if } i \neq \cntrl\\
             \Tilde{d}_{\vrho}^{\vpistar_i}[s] &: \textrm{if } i = \cntrl.
        \end{cases}
    \end{equation*}
\end{restatable}

We see that~\eqref{eq:gengenMVI}---a generalization of \Cref{property:genMVI}---holds without any additional assumptions on the transition probabilities. Yet, decoupling $\Lambda(\vx, \vxstar) \defeq A(\vx) \circ W(\vxstar)$ in the sense of \Cref{property:genMVI} turns out to be crucial to apply our techniques. In fact, the recent hardness result of~\citet{Park23:Multi} suggests that the general case should be intractable. We further remark that \Cref{lemma:property} applies similarly to conclude \Cref{prop:tm-MVI} if we substitute \Cref{assumption:minimax} by \Cref{ass:nonnegativeregret}.

Finally, having established \Cref{lemma:property}, we can now apply \Cref{theorem:genMVI} along with the gradient dominance property (\Cref{lemma:GD}) to obtain one of our main results. Specifically, in \Cref{appendix:proofs} we appropriately bound all of the involved parameters appearing in \Cref{theorem:genMVI}; as usual, this includes a certain \emph{distribution mismatch coefficient} $C_\cG$ (defined in~\eqref{eq:CG} of \Cref{appendix:proofs})---the multi-player analog of the quantity considered by~\citet{Daskalakis20:Independent}---as well as a dependency on $1/\|\vrho\|_\infty$, necessitating that the original distribution $\vrho$ assigns a non-negligible probability mass to all states.

\begin{restatable}{theorem}{main}
    \label{theorem:main}
    Let $\cG$ be a Markov game that satisfies \Cref{ass:nonnegativeregret} and admits a single controller. Then, \eqref{eq:OGD} after $1/\epsilon^2 \cdot \poly(n, \sum_{i=1}^n |\cA_i|, |\cS|, 1/\zeta, C_\cG, 1/\|\vrho\|_\infty)$ iterations computes a stationary $\epsilon$-NE.
\end{restatable}

The importance of \Cref{theorem:genMVI} stems not just from its computational complexity implications, but also from its applicability in a decentralized environment. Indeed, all players are performing gradient steps without any further information from their environment, with the sole exception of the controller. In particular, as predicted by \Cref{lemma:property}, performing the update rule~\eqref{eq:OGD} requires some further access to the environment in order to estimate the (unnormalized) state visitation distribution $\Tilde{d}_{\vrho}^{\vpi}[\cdot]$; using standard arguments, this requires $\poly(|\cG|, 1/\epsilon)$ time to determine within $\epsilon$-error, which suffices for applying \Cref{theorem:genMVI} (see \Cref{remark:errorA}).

An illustrative comparison between~\eqref{eq:OGD} with a time-varying but non-vanishing learning rate---per its update rule---and the vanilla version with a constant learning rate in the context of the ratio game appears in \Cref{appendix:experiments}.



It is worth noting that our proof technique shares an interesting conceptual similarity with the approach of~\citet{Erez22:Regret}, also based on a weighted notion of regret. The key point of departure is that we explicitly incorporate the weights into the update rule~\eqref{eq:OGD}, which in turn induces a second-order dependency on the deviation of the weights in lieu of a first-order bound; this turns out to be crucial for establishing \Cref{theorem:main}. Yet, our approach is more restrictive in that it rests on having a single controller.
\section{Further Related Work}
\label{sec:related}

Computing and learning equilibria in Markov games has attracted considerable interest recently. Most focus has been on the Nash equilibrium in either identical-interest---or more generally, potential---games~\citep{Fox22L:Independent,Leonardos22:Global,Alatur23:Provably,Aydin23:Policy,Ding22:Independent,Zhang22:Global,Maheshwari22:Independent,Macua18:Learning,Chen22:Convergence}, or two-player zero-sum Markov games~\citep{Daskalakis20:Independent,Cen23:Faster,Cai23:Uncoupled,Chen23:Finite,Wei21:Last,Zhang20:Model,Sayin21:Decentralized,Huang22:Towards,Cui22:Offline,Perolat15:Approximate,Zeng22:Regularized,Pattathil23:Symmetric,Yang23:Convergence,Arslantas23:Convergence,Chen23:Twotimescale}, albeit with a few exceptions~\citep{Qin23:Scalable,Sayin23:Decentralized,Giannou22:Convergence,Kalogiannis23:Zero,Kalogiannis23:Efficiently,Park23:Multi,Ma2023:Near}. In general-sum multi-player games, in light of the intractability of Nash equilibria, most focus has been on computing or indeed learning (coarse) correlated equilibria~\citep{Daskalakis22:Complexity,Jin21:V,Wang23:Breaking,Erez22:Regret,Liu22:Sample,Zhang22:Policy,Foster23:Hardness}.

Nevertheless, an important question has been to identify classes of multi-player games that circumvent the intractability of NE in general games. For example, recent work~\citep{Kalogiannis23:Zero,Park23:Multi} investigates the class of polymatrix Markov games, which is based on the homonymous class of normal-form games~\citep{Cai11:Minmax,Cai16:Zero}; indeed, the topic of network games has been particularly popular in the literature on MARL (see~\citep{Zhang18:Fully,Chu20:Multi,Parise19:Graphon}, and references therein). Specifically, \citet{Kalogiannis23:Zero} and \citet{Park23:Multi} leverage the equilibrium collapse of CCE to NE to show that Markov NE can be computed efficiently; in stark contrast, \citet{Park23:Multi} showed that computing a \emph{stationary} NE is $\PPAD$-hard; the latter hardness result is based on earlier work by~\citet{Daskalakis22:Complexity,Jin23:Complexity}. In the class of polymatrix zero-sum Markov games, our novelty compared to earlier work~\citep{Kalogiannis23:Zero,Park23:Multi} (see also the concurrent paper of~\citet{Ma2023:Near}) lies in showing convergence to \emph{stationary} Nash equilibria; this does not contradict the aforementioned hardness results since we impose an additional assumption on the transitions. It is worth underscoring that stationarity is a fundamental desideratum with a long history in the literature on repeated games; among other benefits, stationary policies enjoy a much more memory-efficient encoding, which becomes especially crucial when each policy is represented via an enormous neural network with millions of parameters, while stationary policies are also arguably more interpretable.

Beyond games with separable interactions, \citet{Kalogiannis23:Efficiently} showed that NE can be computed efficiently in a class of games that subsumes both zero-sum and potential games---namely, adversarial team Markov games; see also~\citep{Emmons22:Learning,Wang02:Reinforcement} for pertinent results. It is also worth noting that certain refinements of NE---such as \emph{strict} equilibria---have been shown to be attractors under policy gradient methods~\citep{Giannou22:Convergence}, although such refinements are not universal. 

Naturally, gradient-based methods have also received considerable attention in imperfect-information extensive-form games~\citep{Hoda10:Smoothing,Lee21:Last,Piliouras22:Fast,Zinkevich07:Regret,Liu23:Power}, as well as the more tractable class of normal-form games~\citep{Hsieh21:Adaptive,Hussain23:Asymptotic}. Even for the latter class of games, it is known that gradient-based methods may fail to converge pointwise to Nash equilibria~\citep{Vlatakis20:No,Mertikopoulos18:Cycles}. In stark contrast, it has been documented that \emph{optimism}, a minor modification akin to the extra-gradient method introduced in the online learning literature by~\citet{Rakhlin13:Online,Chiang12:Online}, leads to last-iterate convergence in monotone settings~\citep{Cai22:Finite,Gorbunov22:Last,Golowich20:Last}. Further, beyond the monotone regime, ample of prior work has endeavored to identify broader classes of tractable VIs, such as the \emph{weak} Minty property put forward by~\citet{Diakonikolas21:Efficient}. In turn, this has engendered a considerable recent body of work; we refer to the papers of~\citet{Pethick23:Solving,Cai23:Accelerated,Pethick22:Escaping,Lee21:Fast,Cai22:Accelerated,Mahdavinia22:Tight,Vankov23:Last}, and the many references therein.

Finally, we highlight that Markov games with a single controller have a rich history; see \citep{Parthasarathy81:Orderfield,Bacsar98:Dynamic,Eldosouky16:Single,Guan16:Regret,Qiu21:Provably,Sayin22:Fictitious}; those references contain ample motivation and examples of realistic strategic interactions that can be faithfully modeled as Markov games with a single controller. For example, \citet{Eldosouky16:Single} cast strategically configuring a wireless network so as to protect against potential attacks as a security game in which the defender serves as the sole controller.

\section{Conclusions and Future Work}
\label{sec:conclusion}

In conclusion, we have furnished a natural generalization of the classical Minty property, and we showed that computational tractability persists even under our more permissive condition. We also applied our general theory to obtain new convergence results to stationary Nash equilibria for optimistic policy gradient methods in a broad class of multi-player Markov games.

A number of interesting questions arise from our work. First, our new condition (\Cref{property:genMVI}) crucially relies on the product structure of the joint strategy space. While such structure is always present in multi-player games with uncoupled strategy sets, the canonical case treated in the literature, it may break in some settings of interest~\citep{Jordan23:First,Goktas23:Generative,Goktas22:Exploitability}. Extending the theory of \Cref{sec:beyMinty} to capture such settings is an interesting avenue for future work. Furthermore, we have seen that any Markov game that exhibits equilibrium collapse satisfies property~\eqref{eq:gengenMVI}, without assuming the existence of a single controller. Understanding when property~\eqref{eq:gengenMVI} suffices to ensure computational tractability is another promising direction.

\section*{Acknowledgments}

We are grateful to anonymous reviewers at AAAI for valuable feedback. This material is based on work supported by the Vannevar Bush Faculty
Fellowship ONR N00014-23-1-2876, National Science Foundation grants
RI-2312342 and RI-1901403, ARO award W911NF2210266, and NIH award
A240108S001. Ioannis Panageas would like to acknowledge startup grant from UC Irvine. Part of this work was conducted while Ioannis was visiting Archimedes research unit. This work has been partially supported by project MIS 5154714 of the National Recovery and Resilience Plan Greece 2.0 funded by the European Union under the NextGenerationEU Program.

\bibliography{refs}

\appendix

\section{Additional Preliminaries}
\label{sec:furtherprels}

In this section, we provide some additional preliminaries omitted from the main body.

\paragraph{Average CCE} We first give the definition of an average CCE---in the parlance of~\citet{Nadav10:Limits}. ACCE were also studied recently by~\citet{Zhou23:Stochastic} but under the name ``weak CCE.''

\begin{definition}
    \label{def:acce}
    A stationary (potentially correlated) policy $\vmu \in \Delta(\cA)^\cS$ is a stationary $\epsilon$-ACCE if
    \begin{equation*}
        \sum_{i=1}^n \left( V_i^{\dagger, \vmu_{-i}}(\vrho) - V_i^{\vmu}(\vrho) \right) \leq \epsilon.
    \end{equation*}
\end{definition}

\paragraph{Polymatrix zero-sum Markov games} We recall that a polymatrix zero-sum Markov game is based on an undirected graph $G = (V, E)$. Each node $i \in V \defeq \range{n}$ is uniquely associated with a player. We will denote by $\cN_i \defeq \{i' \in V : \{i, i'\} \in E \}$ the neighborhood of player $i \in \range{n}$. The set of edges $E$ encodes the underlying pairwise interactions, so that the instantaneous reward of player $i \in \range{n}$ can be expressed as $ R_i(s, \va) \defeq \sum_{i' \in \cN_i} R_{i, i'}(s, a_i, a_{i'})$, where $R_{i, i'}: \cS \times \cA_i \times \cA_{i'} \to \R$. Furthermore, it holds that $\sum_{i=1}^n R_i(s, \va) = 0$, for any $(s, \va) \in \cS \times \cA$, since the game is assumed to be zero-sum. Furthermore, \citet{Kalogiannis23:Zero} also assume that in every state $s \in \cS$ there is a single player determining the transition probabilities, denoted by $\argcntrl_s$. Interestingly, admitting a switching controller---in the sense of the latter assumption---is necessary to guarantee equilibrium collapse, as shown by~\citet{Kalogiannis23:Zero}.
\section{Omitted Proofs}
\label{appendix:proofs}

In this section, we provide the proofs omitted from the main body. \Cref{appendix:beyMVI} below contains the proofs from \Cref{sec:beyMinty}, while \Cref{appendix:markov} establishes our statements from \Cref{sec:OGD-Markov}.

\subsection{Proofs from Section~\ref{sec:beyMinty}}
\label{appendix:beyMVI}

We commence here with the proofs from \Cref{sec:beyMinty}. We first have to recall the so-called \emph{RVU property}, crystallized by~\citet{Syrgkanis15:Fast}; the version we include below was not explicitly stated by~\citet{Syrgkanis15:Fast,Rakhlin13:Online}, but follows readily from their arguments.

\begin{proposition}[RVU property~\citep{Syrgkanis15:Fast}]
    \label{prop:RVU}
    Consider a regret minimization algorithm over $\cZ_\ind$ instantiated with \eqref{eq:OGD} parameterized by a learning rate $\eta >0 $. Then, under any sequence of utilities $(\vec{u}^{(t)}_\ind)_{0 \leq t \leq T}$, for some time horizon $T \in \N$, its regret can be upper bounded as
    \begin{align*}
        \reg_{\cZ_\ind}^{(T)} \leq \frac{\diam^2_{\cZ_{\ind}}}{2 \eta} + \eta \sum_{t=1}^T \| \Vec{u}^{(t)}_\ind - \vec{u}^{(t-1)}_{\ind} \|_2^2 - \frac{1}{2\eta} \sum_{t=1}^T \left( \| \vz_\ind^{(t)} - \hvz_\ind^{(t)} \|_2^2 + \|\vz_\ind^{(t)} - \hvz_\ind^{(t+1)} \|_2^2 \right).
    \end{align*}
\end{proposition}

We recall the standard definition of regret: $\reg^{(T)}_{\cZ_\ind} \defeq \max_{\vzstar_\ind \in \cZ_\ind} \sum_{t=1}^T \langle \vzstar_\ind - \vz_\ind^{(t)}, \vec{u}_\ind^{(t)} \rangle$. We also call attention to the fact that~\eqref{eq:OGD} has access to an auxiliary utility $\Vec{u}^{(0)}$, which also appears in the regret bound of \Cref{prop:RVU}; this is just made for convenience, and it does not affect the analysis.

We will also use the following lemma, which can be extracted in~\citep{Anagnostides22:On}.

\begin{lemma}
    \label{lemma:BRgap}
    Suppose that the sequences $(\vz_\ind^{(t)})_{0 \leq t \leq T}$ and $(\hvz_\ind^{(t)})_{1 \leq t \leq T+1}$ are updated by~\eqref{eq:OGD} under a sequence of utilities $(\vec{u}_\ind^{(t)})_{0 \leq t \leq T}$. Then, for any $t \in \range{T}$,
    \begin{equation*}
        \max_{\vzstar_\ind \in \cZ_\ind} \langle \vzstar_\ind - \vz_\ind^{(t)}, \vec{u}_\ind^{(t)} \rangle \leq \left( \frac{\diam_{\cZ_r}}{\eta} + \max_{1 \leq t \leq T} \|\vec{u}_\ind^{(t)} \|_2 \right) s_\ind^{(t)},
    \end{equation*}
    where $s_\ind^{(t)} \defeq \|\vz_\ind^{(t)} - \hvz_\ind^{(t)} \|_2 + \| \vz_\ind^{(t)} - \hvz_\ind^{(t+1)} \|_2$.
\end{lemma}

We can proceed with the proof of \Cref{theorem:genMVI}, the statement of which is recalled below.

\genMVI*

\begin{proof}
    We let $\vx^{(t)} = (\vz_{\ind}^{(t)})_{\ind=1}^d$ for any $t \in \N^*$ and $\hvx^{(t)} = (\hvz_{\ind}^{(t)})_{\ind=1}^d$ for any $t \in \N$. We further let $A = (A_1, \dots, A_d)$ and $F = (F_1, \dots, F_d)$. In light of the Cartesian product structure of $\cX = \bigtimes_{\ind=1}^d \cZ_{\ind}$, the update rule of~\eqref{eq:OGD} can be equivalently written as
    \begin{equation}
    \label{eq:equiv-OGD}
    \begin{split}
    \vz_{\ind}^{(t)} \defeq \proj_{\cZ_{\ind}}(\hvz_{\ind}^{(t)} - \eta A_{\ind}(\vx^{(t-1)}) \circ F_{\ind}(\vx^{(t-1)})),\\
    \hvz_{\ind}^{(t+1)} \defeq \proj_{\cZ_{\ind}}(\hvz_{\ind}^{(t)} - \eta A_{\ind}(\vx^{(t)}) \circ F_{\ind}(\vx^{(t)})),
    \end{split}
\end{equation}
for all times $t \in \N$ and $\ind \in \range{d}$. In turn, \eqref{eq:equiv-OGD} can be equivalently expressed so that $\vz_{\ind}^{(t)}$ and $\hvz_{\ind}^{(t)}$ are solutions to the optimization problems
\begin{equation*}
    \begin{split}
    \min_{\vz_{\ind} \in \cZ_{\ind}}  a_{\ind}(\vx^{(t-1)}) \langle \vz_{\ind}, F_{\ind} (\vx^{(t-1)}) \rangle + \frac{1}{2\eta} \| \vz_{\ind} - \hvz_{\ind}^{(t)} \|_2^2,\\
    \min_{\hvz_{\ind} \in \cZ_{\ind}} a_{\ind}(\vx^{(t)}) \langle \hvz_{\ind}, F_{\ind} (\vx^{(t)}) \rangle + \frac{1}{2\eta} \| \hvz_{\ind} - \hvz_{\ind}^{(t)} \|_2^2,
    \end{split}
\end{equation*}
respectively. Let us fix a time horizon $T \in \N$ and $\ind \in \range{d}$. Invoking \Cref{prop:RVU}, it follows that for any reference point $\vzstar_\ind \in \cZ_{\ind}$ the term $ \sum_{t=1}^T a_{\ind}(\vx^{(t)}) \langle \vz_{\ind}^{(t)} - \vzstar_{\ind}, F_{\ind}(\vx^{(t)}) \rangle$, which can be viewed as the cumulated regret under the sequence of utilities $(F_\ind(\vx^{(t)}) \circ A_\ind(\vx^{(t)}))_{1 \leq t \leq T}$, can be upper bounded by
\begin{align*}
    \frac{\diam_{\cZ_{\ind}}^2}{2\eta} + \eta \sum_{t=1}^T \| a_{\ind}(\vx^{(t)}) F_{\ind}(\vx^{(t)}) - a_{\ind}(\vx^{(t-1)}) &F_{\ind}(\vx^{(t-1)}) \|_2^2 \\
    &- \frac{1}{2\eta} \sum_{t=1}^T \left( \| \vz_{\ind}^{(t)} - \hvz_{\ind}^{(t)} \|_2^2 + \| \vz_{\ind}^{(t)} - \hvz_{\ind}^{(t+1)} \|_2^2 \right),
\end{align*}
where we recall that $\diam_{\cZ_{\ind}}$ above represents the $\ell_2$ diameter of $\cZ_{\ind}$. Furthermore, using the assumption of \Cref{property:genMVI} that $0 < \ell \leq w_\ind(\vxstar) \leq h$, the term $\sum_{t=1}^T w_{\ind}(\vxstar) a_{\ind}(\vx^{(t)}) \langle \vz_{\ind}^{(t)} - \vzstar_{\ind}, F_{\ind}(\vx^{(t)}) \rangle$ can be in turn upper bounded by
\begin{align}
    \frac{\diam_{\cZ_{\ind}}^2 h}{2\eta} + \eta h \sum_{t=1}^T \| a_{\ind}(\vx^{(t)}) F_{\ind}(\vx^{(t)}) - a_{\ind}(\vx^{(t-1)}) &F_{\ind}(\vx^{(t-1)}) \|_2^2 \notag \\
    - &\frac{\ell}{2\eta} \sum_{t=1}^T \left( \| \vz_{\ind}^{(t)} - \hvz_{\ind}^{(t)} \|_2^2 + \| \vz_{\ind}^{(t)} - \hvz_{\ind}^{(t+1)} \|_2^2 \right), \label{align:a-RVU}
\end{align}
for any $\vxstar \in \cX$. Now, by selecting a suitable $\vzstar_{\ind} \in \cZ_\ind$ for each $\ind \in \range{d}$, \Cref{prop:tm-MVI} implies that 
\begin{align*}
    \sum_{t=1}^T \sum_{\ind=1}^d  w_{\ind}(\vxstar) a_{\ind}(\vx^{(t)}) \langle \vz_{\ind}^{(t)} - \vzstar_{\ind}, F_{\ind}(\vx^{(t)}) \rangle =  
    \sum_{t=1}^T \langle \vx^{(t)} - \vxstar, F(\vx^{(t)}) \circ A(\vx^{(t)}) \circ W(\vxstar) \rangle \geq 0,
\end{align*}
for any $t \in \range{T}$. Thus, by adding~\eqref{align:a-RVU} for each $\ind \in \range{d}$ we have that
\begin{align*}
    0 \leq \frac{\diam_{\cX}^2 h}{2\eta} + 2 \eta h^3 \sum_{t=1}^T \| F(\vx^{(t)}) - F(\vx^{(t-1)}) \|_2^2 
    + 2 \eta h \alpha^2 d \max_{1 \leq \ind \leq d} \|F_\ind\|_2^2 \sum_{t=1}^T \| \vx^{(t)} - \vx^{(t-1)} \|_2^2 \\
    - \frac{\ell}{2\eta} \sum_{t=1}^T \left( \| \vx^{(t)} - \hvx^{(t)} \|_2^2 + \| \vx^{(t)} - \hvx^{(t+1)} \|_2^2 \right),
\end{align*}
where we used the fact that $\sum_{\ind=1}^d \diam^2_{\cZ_\ind} = \diam_{\cX}^2$; $\sum_{\ind=1}^d ( \| \vz_{\ind}^{(t)} - \hvz_{\ind}^{(t)} \|_2^2 + \| \vz_{\ind}^{(t)} - \hvz_{\ind}^{(t+1)} \|_2^2 ) = ( \| \vx^{(t)} - \hvx^{(t)} \|_2^2 + \| \vx^{(t)} - \hvx^{(t+1)} \|_2^2 )$; and that 
\begin{align*}
    \| a_{\ind}(\vx^{(t)}) F_{\ind}(\vx^{(t)}) - a_{\ind}(\vx^{(t-1)}) F_{\ind}(\vx^{(t-1)}) \|_2^2 
    &\leq 2 h^2 \| F_{\ind}(\vx^{(t)}) - F_{\ind}(\vx^{(t-1)}) \|_2^2 \\ &+ 2 \|F_{\ind}\|_2^2 | a_{\ind}(\vx^{(t)}) - a_{\ind}(\vx^{(t-1)}) |^2.
\end{align*}
Further, using $L$-Lipschitz continuity of $F$, we have that
\begin{align*}
    0 \leq \frac{\diam_{\cX}^2 h}{2\eta} + 2 \eta (h^3 L^2 + h \BF^2 \alpha^2 d) \sum_{t=1}^T \| \vx^{(t)} - \vx^{(t-1)} \|_2^2
    - \frac{\ell}{8\eta} \sum_{t=1}^T \|\vx^{(t)} - \vx^{(t-1)} \|_2^2\\
    - \frac{\ell}{4\eta} \sum_{t=1}^T \left( \| \vx^{(t)} - \hvx^{(t)} \|_2^2 + \| \vx^{(t)} - \hvx^{(t+1)} \|_2^2 \right),
\end{align*}
where we also used the fact that $\sum_{t=1}^T \|\vx^{(t)} - \vx^{(t-1)} \|_2^2 \leq 2 \sum_{t=1}^T ( \| \vx^{(t)} - \hvx^{(t)} \|_2^2 + \| \vx^{(t)} - \hvx^{(t+1)} \|_2^2 )$. Thus, for $\eta \leq \frac{1}{4} \sqrt{\frac{\ell}{h^3 L^2 + h \BF^2 \alpha^2 d}}$, we conclude the following.

\begin{corollary}
    \label{cor:pathlength}
    Under the conditions of~\Cref{theorem:genMVI},
    \begin{equation*}
    \sum_{t=1}^T \left( \| \vx^{(t)} - \hvx^{(t)} \|_2^2 + \| \vx^{(t)} - \hvx^{(t+1)} \|_2^2 \right) \leq \frac{2 \diam^2_{\cX} h }{\ell}.
\end{equation*}
\end{corollary}
As a result, for $T \geq \frac{2 \diam_{\cX}^2 h}{\ell \epsilon^2}$ the above inequality implies that there exists $t \in \range{T}$ such that $\| \vx^{(t)} - \hvx^{(t)} \|_2, \|\vx^{(t)} - \hvx^{(t+1)} \|_2 \leq \epsilon$. Using \Cref{lemma:BRgap}, it follows that for any $\ind \in \range{d}$,
\begin{align*}
    a_r(\vx^{(t)}) \langle \vz_{\ind}^{(t)}, F_{\ind}(\vx^{(t)})  \rangle &- a_r(\vx^{(t)}) \langle \vzstar_{\ind}, F_{\ind}(\vx^{(t)})  \rangle \\ &\leq \left( \frac{\diam_{\cZ_r}}{\eta} + h \BF \right) \left(\| \vz_{\ind}^{(t)} - \hvz_{\ind}^{(t)} \|_2 + \|\vz_{\ind}^{(t)} - \hvz_{\ind}^{(t+1)} \|_2 \right).
\end{align*}
Given that $a_r(\vx^{(t)}) \geq \ell > 0$,
\begin{align*}
    \langle \vz_{\ind}^{(t)}, F_{\ind}(\vx^{(t)})  \rangle - \langle \vzstar_{\ind}, F_{\ind}(\vx^{(t)})  \rangle \leq 2 \left( \frac{\diam_{\cZ_r}}{\eta \ell} + \frac{h \BF}{\ell} \right) \epsilon.
\end{align*}
Adding those inequalities for all $r \in \range{d}$, we conclude that for $T \geq \frac{2 \diam_{\cX}^2 h}{\ell \epsilon^2}$ it holds that
\begin{equation*}
    \langle \vx^{(t)}, F(\vx^{(t)}) \rangle - \langle \vxstar, F(\vx^{(t)}) \rangle \leq 2 d \left( \frac{ \max_{1 \leq r \leq d} \diam_{\cZ_r}  }{\eta \ell} + \frac{h \BF}{\ell} \right) \epsilon,
\end{equation*}
for any $\vxstar \in \cX$. This concludes the proof.
\end{proof}

\begin{remark}
    \label{remark:errorA}
    \Cref{theorem:genMVI} is robust to conceding error in the evaluation of $A(\vx)$. More precisely, let us suppose that $\| A(\vx^{(t)}) - A'(\vx^{(t)}) \|_\infty \leq \epsilon^{(t)}$, for any $t \in \range{T}^*$ and a sufficiently small $\epsilon^{(t)} > 0$. Following our proof of \Cref{theorem:genMVI}, it is easy to see that~\eqref{eq:OGD} under the sequence $(F(\vx^{(t)}) \circ A'(\vx^{(t)}))_{0 \leq t \leq T}$ yields an $O_T\left( \frac{1}{\sqrt{T}} + \frac{1}{T} \sum_{t=0}^T \| A(\vx^{(t)}) - A'(\vx^{(t)}) \|_\infty \right)$-strong solution to the VI problem. As a result, if we can guarantee that $\epsilon^{(t)} \leq \frac{1}{\sqrt{T}}$, for any $t \in \range{T}^*$, we recover the same rate as \Cref{theorem:genMVI}.
\end{remark}

We next provide a slight extension of \Cref{theorem:genMVI} under a more general condition than \Cref{prop:tm-MVI}, described below.

\begin{property}[Extension of \Cref{prop:tm-MVI}]
    \label{prop:gen-tm-MVI}
    Under the preconditions of \Cref{property:genMVI} with respect to some triple $(\alpha, \ell, h)$, we say that the induced VI problem satisfies the average $(\alpha, \ell, h) \in \R^3_{> 0}$-generalized Minty property with slackness $\gamma > 0$ if for any sequence $\sigma^{(T)} \defeq (\vx^{(t)})_{1 \leq t \leq T}$ there exists $\cX \ni \vxstar = \vxstar(\sigma^{(T)})$ so that
    \begin{equation*}
        \frac{1}{T} \sum_{t=1}^T \langle \vx^{(t)} - \vxstar, F(\vx^{(t)}) \circ A(\vx^{(t)}) \circ W(\vxstar) \rangle \geq -\gamma.
    \end{equation*}
\end{property}

\begin{corollary}
    \label{cor:gammagenMVI}
    Let $\cX = \bigtimes_{\ind=1}^d \cZ_\ind$ for some $d \in \N$ and $F : \cX \to \cX$ be an $L$-Lipschitz continuous operator. Suppose further that the average $(\alpha, \ell, h)$-generalized Minty property with slackness $\gamma > 0$ (\Cref{prop:gen-tm-MVI}) holds. Then, for any $\epsilon > 0$, after $T \in \N$ iterations of \eqref{eq:OGD} with learning rate $\eta \leq \frac{1}{4} \sqrt{\frac{\ell}{h^3 L^2 + h \BF^2 \alpha^2 d}}$ there is a point $\vx^{(t)} \in \cX$ such that for any $\vxstar \in \cX$,
    \begin{equation*}
        \langle \vx^{(t)} - \vxstar, F(\vx^{(t)}) \rangle \leq 2 d \left( \frac{\max_{1 \leq r \leq d} \diam_{\cZ_r}}{\eta \ell} + \frac{h \BF}{\ell} \right) \sqrt{\frac{4\eta \gamma}{\ell} + \frac{2\diam_{\cX}^2 h}{\ell T}}.
    \end{equation*}
\end{corollary}

\begin{proof}
    Similarly to the proof of \Cref{cor:pathlength}, $\sum_{t=1}^T \left( \| \vx^{(t)} - \hvx^{(t)} \|_2^2 + \| \vx^{(t)} - \hvx^{(t+1)} \|_2^2 \right) \leq \frac{2 \diam^2_{\cX} h }{\ell} + \frac{4\eta \gamma}{\ell} T$. As a result, we conclude that there is $t \in \range{T}$ such that 
    \begin{equation*}
        \| \vx^{(t)} - \hvx^{(t)} \|_2, \|\vx^{(t)} - \hvx^{(t+1)} \|_2 \leq \sqrt{ \frac{2\diam_\cX^2 h}{\ell T} + \frac{4\eta \gamma}{\ell}}.
    \end{equation*}
    The statement then follows from \Cref{lemma:BRgap}, similarly to \Cref{theorem:genMVI}.
\end{proof}

\begin{remark}
    One application of incorporating slackness---per \Cref{prop:gen-tm-MVI} and variants thereof---with independent interest pertains the convergence of competing neural networks in the neural tangent kernel (NTK) regime~\citep{Jacot18:Neural}, wherein the optimization landscape behaves as nearly convex-concave. Indeed, online learning techniques readily extend under near convexity~\citep{Chen23:Regret}, a fact that can be leveraged in conjunction with our approach to provide convergence guarantees in that regime.
\end{remark}

We next extend our analysis in the presence of noise in the operator, in a sense that will be made precise very shortly. To this end, we first need to state a slight modification of \Cref{prop:gen-tm-MVI}; we recall the notation $\eqgap(\vx) : \cX \ni \vx \mapsto \max_{\vxstar \in \cX} \langle \vx - \vxstar, F(\vx) \rangle$.

\begin{property}
    \label{prop:gengen-tm-MVI}
    Under the preconditions of \Cref{property:genMVI} with respect to some triple $(\alpha, \ell, h)$, we say that the induced VI problem satisfies the average $(\alpha, \ell, h, \rho) \in \R^4_{> 0}$-generalized Minty property if for any sequence $\sigma^{(T)} \defeq (\vx^{(t)})_{1 \leq t \leq T}$ there exists $\cX \ni \vxstar = \vxstar(\sigma^{(T)})$ so that
    \begin{align*}
        \sum_{t=1}^T \langle \vx^{(t)} - \vxstar, F(\vx^{(t)}) \circ A(\vx^{(t)}) \circ W(\vxstar) \rangle \geq
        - \rho \sum_{t=1}^T (\eqgap(\vx^{(t)}))^2.
    \end{align*}
\end{property}

\begin{corollary}
    \label{cor:noisy}
    Let $\cX = \bigtimes_{\ind=1}^d \cZ_\ind$ for some $d \in \N$ and $F : \cX \to \cX$ be an $L$-Lipschitz continuous operator. Suppose that we instead have access to an operator $\noisyF : \cX \to \cX$ such that
    \begin{enumerate}
        \item $\noisyF$ satisfies the average $(\alpha, \ell, h, \rho)$-generalized Minty property (\Cref{prop:gengen-tm-MVI}); and
        \item $\| \noisyF(\vx) - F(\vx) \|_2 \leq \delta \cdot \eqgap(\vx)$, for any $\vx \in \cX$.
    \end{enumerate}
    If the pair $(\rho, \delta)$ is such that 
    \begin{equation*}
        \rho + 12 \eta h^3 \delta^2 \leq  \frac{\eta \ell^3}{64 (\max_{1 \leq r \leq d} \diam_{\cZ_r} + \eta h \BFnoisy)^2}
    \end{equation*}
    and $\delta \leq \frac{1}{2 \diam_{\cX}}$, with learning rate $\eta \leq \frac{1}{4} \sqrt{ \frac{\ell}{ 3h^3 L^2 + h \BFnoisy^2 \alpha^2 d}}$, then after $T \in \N$ iterations of~\eqref{eq:OGD} there is a point $\vx^{(t)} \in \cX$ with equilibrium gap $\eqgap(\vx^{(t)})$ upper bounded by
    \begin{equation*}
         \frac{16 \eta K^2}{\sqrt{T} \ell} \left( 6 \eta h^3 \delta^2 \eqgap(\vx^{(0)}) + \frac{\diam_{\cX}^2 h}{2\eta} \right),
    \end{equation*}
    for any $\epsilon > 0$, where $K$ is defined in~\eqref{eq:K}.
\end{corollary}

\begin{proof}
    First, using \Cref{lemma:BRgap}, it follows that for any $\ind \in \range{d}$,
    \begin{align*}
        a_r(\vx^{(t)}) \langle \vz_{\ind}^{(t)}, \noisyF_{\ind}(\vx^{(t)})  \rangle &- a_r(\vx^{(t)}) \langle \vzstar_{\ind}, \noisyF_{\ind}(\vx^{(t)})  \rangle \\ &\leq \left( \frac{\max_{1 \leq r \leq d} \diam_{\cZ_r} }{\eta} + h \BFnoisy \right) \left(\| \vz_{\ind}^{(t)} - \hvz_{\ind}^{(t)} \|_2 + \|\vz_{\ind}^{(t)} - \hvz_{\ind}^{(t+1)} \|_2 \right),
    \end{align*}
    in turn implying that
    \begin{align*}
        \langle \vx^{(t)} - \vxstar, \noisyF(\vx^{(t)}) \rangle \leq \left( \frac{\max_{1 \leq r \leq d} \diam_{\cZ_r}}{\eta \ell} + \frac{h \BFnoisy}{\ell} \right) \left(\| \vx^{(t)} - \hvx^{(t)} \|_2 + \|\vx^{(t)} - \hvx^{(t+1)} \|_2 \right).
    \end{align*}
    Given that, by assumption, $\| \noisyF(\vx^{(t)}) - F(\vx^{(t)}) \|_2 \leq \delta \eqgap(\vx^{(t)})$, it follows that for $\delta \leq \frac{1}{2 \diam_\cX}$,
    \begin{equation*}
        \eqgap(\vx^{(t)}) \leq K \left(\| \vx^{(t)} - \hvx^{(t)} \|_2 + \|\vx^{(t)} - \hvx^{(t+1)} \|_2 \right),
    \end{equation*}
    where we have defined
    \begin{equation}
        \label{eq:K}
        K \defeq 2 \left( \frac{\max_{1 \leq r \leq d} \diam_{\cZ_r}}{\eta \ell} + \frac{h \BFnoisy}{\ell} \right).
    \end{equation}
    Moreover, similarly to the proof of \Cref{theorem:genMVI}, \Cref{prop:gengen-tm-MVI} implies that the term $-\rho \sum_{t=1}^T (\eqgap(\vx^{(t)}))^2$ can be upper bounded by
    \begin{align*}
     \frac{\diam_{\cX}^2 h}{2\eta} + 2 \eta h^3 \sum_{t=1}^T \| \noisyF(\vx^{(t)}) - \noisyF(\vx^{(t-1)}) \|_2^2
    + 2 \eta h \alpha^2 d \BFnoisy^2 \sum_{t=1}^T \| \vx^{(t)} - \vx^{(t-1)} \|_2^2 \\
    - \frac{\ell}{2\eta} \sum_{t=1}^T \left( \| \vx^{(t)} - \hvx^{(t)} \|_2^2 + \| \vx^{(t)} - \hvx^{(t+1)} \|_2^2 \right).
\end{align*}
    Now the term $- \frac{\ell}{2\eta} \sum_{t=1}^T \left( \| \vx^{(t)} - \hvx^{(t)} \|_2^2 + \| \vx^{(t)} - \hvx^{(t+1)} \|_2^2 \right)$ can be upper bounded by
    \begin{align*}
        -\frac{\ell}{8\eta} \sum_{t=1}^T \|\vx^{(t)} - \vx^{(t-1)} \|_2^2 - \frac{\ell}{8\eta K^2} \sum_{t=1}^T (\eqgap^{(t)})^2,
    \end{align*}
    while the term $\sum_{t=1}^T \| \noisyF(\vx^{(t)}) - \noisyF(\vx^{(t-1)}) \|_2^2$ can be upper bounded by
    \begin{align*}
        3 \delta^2 \eqgap(\vx^{(0)})^2 + 6\delta^2 \sum_{t=1}^T (\eqgap(\vx^{(t)}))^2 + 3L^2 \sum_{t=1}^T \|\vx^{(t)} - \vx^{(t-1)} \|_2^2.
    \end{align*}
    As a result, for $\eta \leq \frac{1}{4} \sqrt{ \frac{\ell}{ 3h^3 L^2 + h \BFnoisy^2 \alpha^2 d}}$ and any pair $(\rho, \delta)$ such that
    \begin{equation*}
        \rho + 12 \eta h^3 \delta^2 \leq \frac{\ell}{16\eta K^2},
    \end{equation*}
    we conclude that there is $\vx^{(t)} \in \cX$ with equilibrium gap $\eqgap(\vx^{(t)})$ upper bounded by
    \begin{equation*}
         \frac{16 \eta K^2}{\sqrt{T} \ell} \left( 6 \eta h^3 \delta^2 \eqgap(\vx^{(0)}) + \frac{\diam_{\cX}^2 h}{2\eta} \right).
    \end{equation*}
\end{proof}

\subsection{Proofs from Section~\ref{sec:OGD-Markov}}
\label{appendix:markov}

In this section, we provide the proofs deferred from \Cref{sec:OGD-Markov}. We first make a remark regarding differentiability of the value function, following~\citep[Remark 1]{Daskalakis20:Independent}. 

\begin{remark}[Differentiability]
    \label{remark:diff}
    Under direct parameterization, the interior of the joint strategy space, denoted by $\inter(\cX)$, is empty. To make sure that the gradient $\nabla_{\vx_i} V_i(\vrho)$ is well-defined, we can instead consider a suitable compact and convex set $\cX_\delta$, for any $\delta > 0$, so that $\cX \subseteq \inter(\cX_\delta)$ and any point $\cX$ is within distance $\delta$ from some point in $\cX_\delta$. Using continuity and compactness, it is direct to see that by taking the limit $\delta \downarrow 0$ our analysis readily applies.
\end{remark}

Now, following the approach of~\citet{Kalogiannis23:Zero}, we prove \Cref{prop:polymatrix-mvi}.

\polymatrixmvi*

\begin{proof}
    We consider for each player $i \in \range{n}$ the following nonlinear program with variables $\vv_i \in \R^\cS$ and $\vmu \in \Delta(\cA)^\cS$.
    \begin{align*}
        &\min \vrho^\top \vv_i \\
        &\text{ s.t. } \vv_i[s] \geq \E_{\va_{-i} \sim \vmu_{-i}(\cdot | s) } [R_i(s, \vec{a}) + \bar{\zeta}_{s, \va} \pr(\cdot | s, \va) \vv_i],
    \end{align*}
    where $\bar{\zeta}_{s, \va} \defeq 1- \zeta_{s, \va}$, for all $s \in \cS$ and $a_i \in \cA_i$. In particular, $\vmu$ above represents a stationary, potentially correlated joint policy. Now let us fix a player $i \in \range{n}$ and $\vmu_{-i} \in \Delta(\cA_{-i})^\cS$. It is well-known that the induced linear program is feasible, and the optimal objective is equal to the value of player $i \in \range{n}$ when best responding to $\vmu_{-i}$~\citep{Puterman05:Markov}. That is, if the optimal value is attained at $\vvstar_i(\vmu_{-i}) \in \R^\cS$, it holds that $V_i^{\dagger, \vmu_{-i}}(\vrho) = \vrho^\top \vvstar_i(\vmu_{-i})$, for all $i \in \range{n}$. In particular, if $\vmu(\cdot | s) \coloneqq \vpistar(\cdot | s)$ is a Nash equilibrium policy (\Cref{def:NE}), it holds that $\sum_{i=1}^n V_i^{\dagger, \vpistar_{-i}}(\vrho) = \sum_{i=1}^n V_i^{\vpistar}(\vrho) = 0$. This in turn implies that the sum of the objectives (over the players) of the original nonlinear programs is nonpositive. Furthermore, let us fix $\vpi(\cdot | s)$ to be a product distribution. By feasibility, it follows that for any state $s \in \cS$,
    \begin{equation*}
        \vv_i[s] \geq \E_{\va \sim \vpi(\cdot | s)} [R_i(s, \va) + \bar{\zeta}_{s, \va} \pr(\cdot |s, \va) \vv_i].
    \end{equation*}
    Thus, using the fact that $\sum_{i=1}^n R_i(s, \va) = 0$ and that $\zeta_{s, \va} > 0$ for all $(s, \va) \in \cS \times \cA$, it follows that $\sum_{i=1}^n \vv_i[s] \geq 0$ for all $s \in \cS$. So, it follows that when restricting $\vmu$ to be a product policy the sum of the objective values is $0$.

    Now, consider a potentially correlated policy $\vmu$, and let $\vpi = \vpi(\vmu)$ be the product distribution induced by taking the marginals of $\vmu$. By feasibility, for any player $i \in \range{n}$ and $(s, a_i) \in \cS \times \cA_i$,
   \begin{align}
        \label{align:feas}
       \vvstar_i[s] \geq \E_{\va_{-i} \sim \vmu_{-i}(\cdot | s) } [R_i(s, \vec{a}) + \bar{\zeta}_{s, \va} \pr(\cdot | s, \va) \vvstar_i].
   \end{align} 
    By the assumption of separability of the reward function, the first term in the right-hand side of~\eqref{align:feas} is equal to 
    \begin{equation*}
        \sum_{i' \in \cN_i} \E_{\va_{-i} \sim \vmu_{-i}(\cdot | s) } R_{i, i'}(s, \va) = \E_{\va_{-i} \sim \vpi_{-i}(\cdot | s) } R_i(s, \va).
    \end{equation*} 
    Further, by the assumption of having a switching controller, the second term in the right-hand side of~\eqref{align:feas} is equal to 
    \begin{equation*}
        \E_{\va_{-i} \sim \vmu_{-i}(\cdot | s) } \pr(\cdot | s, \va) \vvstar_i = \E_{\va_{-i} \sim \vpi_{-i}(\cdot | s) } \pr(\cdot | s, \va) \vvstar_i.
    \end{equation*}
    Thus, we conclude that
    \begin{equation*}
    \vvstar_i[s] \geq \E_{\va_{-i} \sim \vpi_{-i}(\cdot | s) } [R_i(s, \vec{a}) + \bar{\zeta}_{s, \va} \pr(\cdot | s, \va) \vvstar_i],    
    \end{equation*}
    which means that the pair $(\vvstar_i, \vpi(\vmu))$ constitutes a feasible solution. Given that $\vpi(\vmu)$ is by definition a product distribution, we know that $\sum_{i=1}^n \vrho^\top \vvstar_i \geq 0$, in turn implying that $\sum_{i=1}^n V_i^{\dagger, \vmu_{-i}}(\vrho) \geq 0$. Finally, \Cref{ass:nonnegativeregret} follows by taking $\vmu$ to be a uniform mixture of $T$ product distributions.
\end{proof}

Beyond polymatrix zero-sum Markov games, \Cref{ass:nonnegativeregret} is satisfied in all games exhibiting equilibrium collapse per \Cref{def:equilcol}, as we observe next.

\nonneg*

\begin{proof}
    For the sake of contradiction, suppose that there exists a sequence of product joint policies $(\vpi^{(1)}, \dots, \vpi^{(T)})$ such that
    \begin{equation}
        \label{eq:neg}
        \frac{1}{T} \sum_{t=1}^T \sum_{i=1}^n V_{i}^{\dagger, \vpi^{(t)}_{-i}}(\vrho) - \frac{1}{T} \sum_{t=1}^T \sum_{i=1}^n V_{i}^{\vpi^{(t)}}(\vrho) < 0;
    \end{equation}
    that is, \Cref{ass:nonnegativeregret} is violated. If $\vmu \in \Delta(\cA)^\cS$ represents the uniform mixture over $(\vpi^{(1)}, \dots, \vpi^{(T)})$, which is potentially a correlated policy, \eqref{eq:neg} can be rewritten as
    \begin{equation*}
        \sum_{i=1}^n \left( V_i^{\dagger, \vmu_{-i}}(\vrho) - V_i^{\vmu}(\vrho) \right) < 0.
    \end{equation*}
    In words, $\vmu$ constitutes an $\epsilon$-ACCE (\Cref{def:acce}) with $\epsilon < 0$. By the assumption that $\cG$ exhibits equilibrium collapse (\Cref{def:equilcol}), it follows that the marginals of $\vmu$ induce an $\epsilon'$-Nash equilibrium with $\epsilon' < 0$, which is a contradiction. This completes the proof.
\end{proof}

We next state a number of elementary properties in MDPs~\citep{Cai20:Provably}, starting from the connection between the gradient of the value function and the $Q$ function; for completeness, we also provide their proofs.

\begin{lemma}
    \label{lemma:Q-V}
    For any state $s \in \cS$ and any joint action profile $(a_i, \va_{-i}) = \vec{a} \in \cA$,
    \begin{equation*}
        \frac{\partial V^{\vpi}_{i}(\vrho)}{\partial \vx_{i, s}[a_i]} = \Tilde{d}_{\vrho}^{\vpi}[s] \E_{\va_{-i} \sim \vpi_{-i}(\cdot | s)} [ Q_i^{\vpi}(s, \va)], \forall i \in \range{n}.
    \end{equation*}
\end{lemma}

\begin{proof}
    Let $s_0 \in \cS$ be any initial state. We have that the gradient $\nabla_{\vx_i} V_i^{\vpi}(s_0)$ is equal to
    \begin{equation*}
        \nabla_{\vx_i} \left( \sum_{a_{i, 0} \in \cA_i} \vpi_i(a_{i, 0} | s_0) \E_{\va_{-i, 0} \sim \vpi_{-i}(\cdot | s_0)} [Q_i^{\vpi}(s_0, \va_{0})] \right),
    \end{equation*}
    where we denoted by $\va_0 \defeq (a_{1, 0}, \dots, a_{n, 0})$. The above display is in turn equal to
    \begin{align*}
        \sum_{a_{i, 0} \in \cA_i} \vpi_i(a_{i, 0} | s_0) &(\nabla_{\vx_i} \log  \vpi_i(a_{i, 0} | s_0)) \E [Q_i^{\vpi}(s_0, \va_{0})] \\
        &+ \sum_{a_{i, 0} \in \cA_i} \vpi_i(a_{i, 0} | s_0) \E\left[ \sum_{s_1 \in \cS} \pr(s_1 |s_0, \va_0) \nabla_{\vx_i} V_i^{\vpi}(s_1)  \right],
    \end{align*}
    where the expectation is taken with respect to $\va_{-i, 0} \sim \vpi_{-i}(\cdot | s_0)$. In particular, the second term above follows from the fact that
    \begin{equation*}
        Q_i^{\vpi}(s_0, \va_0) = \sum_{s_1 \in \cS} \pr(s_1 | s_0, \va_0) V_{i}^{\vpi}(s_1).
    \end{equation*}
    As a result, it follows that the gradient $\nabla_{\vx_i} V_i^{\vpi}(s_0)$ can be expressed as
    \begin{align*}
        \E_{\pr^{\vpi}(\cdot | s_0)}[\nabla_{\vx_i} \log  \vpi_i(a_{i, 0} | s_0)) Q_i^{\vpi}(s_0, \va_0)] 
        + \E_{\pr^{\vpi}(\cdot | s_0)} \left[\mathbbm{1}_{H \geq 1} \nabla_{\vx_i} V^{\vpi}_i(s_0) \right].
    \end{align*}
    By linearity, the same holds by replacing the initial state $s_0 \in \cS$ with any distribution $\vrho \in \Delta(\cS)$. As a result, by induction and the fact that $\zeta > 0$, we conclude that
    \begin{equation*}
        \nabla_{\vx_i} V_i^{\vpi}(\vrho) = \E \left[ \sum_{h=0}^H \nabla_{\vx_i} \log  \vpi_i(a_{i, h} | s_h)) Q_i^{\vpi}(s_h, \va_h)  \right].
    \end{equation*}
    The above expression is also equal to
    \begin{equation*}
        \sum_{s \in \cS} \Tilde{d}_{\vrho}^{\vpi}[s] \E_{\va \sim \vpi(\cdot | s)} [ \nabla_{\vx_i} \log \vpi_i(a_i | s)) Q_i^{\vpi}(s, \va) ].
    \end{equation*}
    The statement of the lemma thus follows by the fact that $\frac{\partial \log \vpi_i(a_i | s))}{\partial \vx_{i, s'}[a_i']} = \frac{1}{\vx_{i, s}[a_i]}$ if $(s, a_i) = (s', a_i')$, and $0$ otherwise.
\end{proof}

\begin{lemma}[Value difference]
    \label{lemma:Valdiff}
    For any joint policy $\vpi \in \Pi$ and policy $\vpi_i' \in \Pi_i$, the value difference $V_i^{\vpi_i', \vpi_{-i}}(\vrho) - V_i^{\vpi}(\vrho) $ is equal to
    \begin{equation*}
        \sum_{s \in \cS} \Tilde{d}_{\vrho}^{\vpi_i', \vpi_{-i}}[s] \sum_{a_i \in \cA_i} ( \vx_{i, s}'[a_i] - \vx_{i, s}[a_i] ) \E [Q_i^{\vpi} (s, \va)],
    \end{equation*}
    for any player $i \in \range{n}$, where the expectation above is taken over $\va_{-i} \sim \vpi_{-i}(\cdot | s)$. 
\end{lemma}

\begin{proof}
    Let $s \in \cS$. We see that the value difference $ V_i^{\vpi_i', \vpi_{-i}}(s) - V_i^{\vpi}(s)$ can be expressed as
    \begin{align*}
        \E \left[ \sum_{h=0}^H R_i(s_h, \va_h) \right] - V_i^{\vpi}(s) 
        &= \E \left[ \sum_{h=0}^H R_i(s_h, \va_h) + \mathbbm{1}_{h+1 \leq H} V_i^{\vpi}(s_h) - V_i^{\vpi}(s_h) \right] \\
        &= \E \left[ \sum_{h=0}^H \left( Q_i^{\vpi}(s_h, \va_h) - V_i^{\vpi}(s_h) \right) \right],
    \end{align*}
    where the expectation above is taken over trajectories induced by $\pr^{\vpi_i', \vpi_{-i}}$. As a result, the value difference $ V_i^{\vpi_i', \vpi_{-i}}(s) - V_i^{\vpi}(s)$ is equal to
    \begin{equation*}
        \sum_{s' \in \cS} \Tilde{d}_{s}^{\vpi_i', \vpi_{-i}}[s'] \E_{\va \sim (\vpi_i'(\cdot | s'), \vpi_{-i}(\cdot | s')) } \left[ Q_i^{\vpi}(s', \va) - V_i^{\vpi}(s') \right],
    \end{equation*}
    which leads to the conclusion of the lemma by taking the expectation $\E_{s \sim \vrho} [\cdot]$.
\end{proof}

We now combine \Cref{lemma:Q-V,lemma:Valdiff} to conclude \Cref{lemma:property}, the statement of which is recalled below.

\prop*

\begin{proof}
    Let us consider a player $i \in \range{n}$. For $(s, a_i) \in \cS \times \cA_i$, \Cref{lemma:Q-V} implies that
    \begin{equation*}
        \Lambda_i(\vx, \vxstar)[s, a_i] \frac{\partial V^{\vpi}_{i}(\vrho)}{\partial \vx_{i, s}[a_i]} = \Tilde{d}_{\vrho}^{\vpistar_i, \vpi_{-i}}(s) \E [ Q_i^{\vpi}(s, \va)],
    \end{equation*}
    where the expectation above is taken over $\va_{-i} \sim \vpi_{-i}(\cdot | s)$. Thus, summing over all $s \in \cS$ and $a_i \in \cA_i$, it follows that the term $\langle \vxstar_i - \vx_i, 
\nabla_{\vx_i} V_i(\vrho) \circ \Lambda_i(\vx, \vxstar) \rangle $ is equal to
    \begin{equation*}
        \sum_{s \in \cS} \Tilde{d}_{\vrho}^{\vpi_i', \vpi_{-i}}[s] \sum_{a_i \in \cA_i} ( \vx_{i, s}'[a_i] - \vx_{i, s}[a_i] ) \E [Q_i^{\vpi} (s, \va)],
    \end{equation*}
    where the expectation above is again taken over $\va_{-i} \sim \vpi_{-i}(\cdot | s)$. By \Cref{lemma:Valdiff}, the term above is equal to the value difference $V_i^{\vpi_i', \vpi_{-i}}(\vrho) - V_i^{\vpi}(\vrho)$. Thus, summing over all players $i \in \range{n}$, we find that \eqref{eq:genMVI} is a consequence of \Cref{assumption:minimax}. Analogously, \eqref{eq:tm-MVI} is a consequence of \Cref{ass:nonnegativeregret}.
\end{proof}

\begin{remark}[Greedy exploration]
    \label{remark:exploration}
    Throughout this paper, we have been operating in the regime of direct parameterization in that $\vpi_i(a_i | s) \defeq \vx_{i, s}[a_i]$, for any player $i \in \range{n}$ and $(s, a_i) \in \cX \times \cA_i$. This type of parameterization suffices under the assumption that players have complete gradient feedback, but in the more challenging bandit feedback model such a parameterization could cause the variance of the gradient estimator to blow up. One common approach to address this issue consists of incorporating $\gamma$-greedy exploration, so that now $\vpi_i(a_i | s) \defeq (1- \gamma) \vx_{i, s}[a_i] + \gamma/|\cA_i|$. This in turn leads to variance bounded by $O_\gamma(1/\gamma)$. It is fairly straightforward to see that \Cref{lemma:property} still implies \Cref{prop:tm-MVI}, with the difference that the right-hand side of~\eqref{eq:tm-MVI} is replaced by a term $- \Theta_\gamma(\gamma) T$. By virtue of \Cref{cor:gammagenMVI}, analogous conclusions hold in that case as well.
\end{remark}

In \Cref{cor:local}, we saw that the guarantee of \Cref{theorem:genMVI} in general smooth multi-player games yields only a local optimality guarantee; to arrive at global optimality, as claimed in \Cref{theorem:main}, we will show that the gradient dominance property (\Cref{item:GD}) holds; the proof below follows that of~\citep[Lemma 1]{Daskalakis20:Independent}.

\begin{lemma}[Gradient dominance]
    \label{lemma:GD}
    Let $\vpi \in \Pi$ and $\vpi_i' \in \Pi_i$, for some player $i \in \range{n}$. Then, the value difference $\max_{\vpi_i' \in \Pi_i} V_i^{\vpi_i', \vpi_{-i}}(\vrho) - V_i^{\vpi}(\vrho)$ is upper bounded by
    \begin{equation*}
        \min_{\vpistar_i \in \Pistar_{i}(\vpi_{-i})} \left\| \frac{d_{\vrho}^{\vpistar_i, \vpi_{-i}}}{\vrho} \right\|_\infty \frac{1}{\zeta} \max_{\vx'_i \in \cX_i} \langle \vx_i' - \vx, \nabla_{\vx_i} V_i(\vrho) \rangle.
    \end{equation*}
\end{lemma}

Before we proceed with the proof, let us point out that the ratio above is defined coordinate-wise, which is well-defined since we have assumed that $\vrho$ has full support. We also note that the nomenclature $\Pistar_i(\vpi_{-i})$ above denotes the set of (stationary) best response policies of player $i \in \range{n}$ under $\vpi_{-i}$.

\begin{proof}[Proof of \Cref{lemma:GD}]
    By \Cref{lemma:Valdiff}, we have that the value difference $\max_{\vpi_i' \in \Pi_i} V_i^{\vpi_i', \vpi_{-i}}(\vrho) - V_i^{\vpi}(\vrho)$ can be upper bounded by 
    \begin{align*}
        \sum_{s \in \cS} \Tilde{d}_{\vrho}^{\vpistar_i, \vpi_{-i}}[s] \max_{a_i \in \cA_i} \E_{\va_{-i} \sim \vpi_{-i}(\cdot | s) } [V_i^{\vpi}(s) - Q_i^{\vpi}(s, \va)],
    \end{align*}
    where $\vpistar_i(\vpi_{-i}) \in \Pistar_{i}(\vpi_{-i})$ is a policy minimizing $\left\| \frac{d_{\vrho}^{\vpistar_i, \vpi_{-i}}}{\vrho} \right\|_\infty$. Since $\max_{a_i \in \cA_i} \E_{\va_{-i} \sim \vpi_{-i}(\cdot | s) } [V_i^{\vpi}(s) - Q_i^{\vpi}(s, \va)] \geq 0$ for any state $s \in \cS$, the last displayed term can be in turn upper bounded by 
    \begin{equation*}
        \left\| \frac{\tilde{d}_{\vrho}^{\vpistar_i, \vpi_{-i}}}{\tilde{d}_{\vrho}^{\vpi}} \right\|_\infty \sum_{s \in \cS} \Tilde{d}_{\vrho}^{\vpi}[s] \max_{a_i \in \cA_i} \E_{\va_{-i} \sim \vpi_{-i}(\cdot | s) } [V_i^{\vpi}(s) - Q_i^{\vpi}(s, \va)].
    \end{equation*}
     Moreover, the first term above can be bounded as
     \begin{equation*}
         \left\| \frac{\tilde{d}_{\vrho}^{\vpistar_i, \vpi_{-i}}}{\tilde{d}_{\vrho}^{\vpi}} \right\|_\infty \leq \frac{1}{\zeta} \left\| \frac{d_{\vrho}^{\vpistar_i, \vpi_{-i}}}{\vrho} \right\|_\infty,
     \end{equation*}
     while the second term is equal to
     \begin{align*}
         \max_{\vx_i' \in \cX_i} \sum_{s \in \cS} \sum_{a_i \in \cA_i} &\Tilde{d}_{\vrho}^{\vpi}[s] \vx_{i, s}'[a_i] \E [V_i^{\vpi}(s) - Q_i^{\vpi}(s, \va)] \\
         &= \max_{\vx_i' \in \cX_i} \sum_{s \in \cS} \sum_{a_i \in \cA_i} \Tilde{d}_{\vrho}^{\vpi}[s] (\vx_{i, s}[a_i] - \vx_{i, s}'[a_i]) \E [Q_i^{\vpi}(s, \va)],
     \end{align*}
     where the expectation is taken over $\va_{-i} \sim \vpi_{-i}(\cdot | s)$. By \Cref{lemma:Q-V}, the last term can be recognized as $\max_{\vx_i' \in \cX_i} \langle \vx_i - \vx_i', \nabla_{\vx_i} V_i(\vrho) \rangle$, concluding the proof.
\end{proof}

Finally, to conclude \Cref{theorem:genMVI} using \Cref{theorem:main}, we appropriately bound all the involved parameters. We first point out a standard bound on the smoothness of the value function.

\begin{lemma}
    \label{lemma:smooth}
    For any joint policies $\vpi, \vpi' \in \Pi$, 
    \begin{equation*}
        \| \nabla_{\vx_i} V^{\vpi}_i(\vrho) - \nabla_{\vx_i} V^{\vpi'}_i(\vrho)  \|_2 \leq \frac{4 |\cA_i|}{\zeta^3} \| \vx - \vx' \|_2,
    \end{equation*}
    for any player $i \in \range{n}$.
\end{lemma}

\main*

\begin{proof}
    In light of \Cref{lemma:property}, we will apply \Cref{theorem:genMVI} with the following parameters:

    \begin{itemize}
        \item $\cX_i \defeq \Delta(\cA_i)^\cS$ and $\cX \defeq \bigtimes_{i=1}^n \cX_i$. As such, we have that $\cX = \bigtimes_{i \in \range{n}, s \in \cS} \cZ_{i, s}$ with $\cZ_{i, s} = \Delta(\cA_i)$;
        \item $d \defeq n |\cS|$;
        \item $\diam_{\cX}^2 = 2 n |\cS|$;
        \item $h \defeq \max \left\{ \frac{1}{\zeta}, \frac{1}{\| \vrho \|_\infty} \right\}$ and $\ell \defeq \min \left\{ \zeta, \|\vrho\|_\infty \right\}$. This follows given that for any joint policy $\vpi \in \Pi$ it holds that $\tilde{d}^{\vpi}_{s_0}[s] = \sum_{h \in \N^*} \pr^{\vpi}(s_h = s | s_0) \leq \sum_{h=0}^\infty (1 - \zeta)^h \leq \frac{1}{\zeta}$ and that $\tilde{d}^{\vpi}_{s_0}[s] \geq \pr^{\vpi}(s_0 = s | s_0)$, for any $(s, s_0) \in \cS \times \cS$, in turn implying that $\tilde{d}^{\vpi}_{\vrho}[s] \leq \frac{1}{\zeta}$ and $\tilde{d}^{\vpi}_{\vrho}[s] \geq \vrho[s]$, for any $s \in \cS$. Hence, the claimed bounds on $\ell$ and $h$ follow directly by virtue of \Cref{lemma:property};  
        \item $L \defeq \frac{4 \sqrt{\sum_{i=1}^n |\cA_i|^2}}{\zeta^3}$. Indeed, having taken $F(\vx) \defeq - (\nabla_{\vx_1} V_{1}(\vrho), \dots, \nabla_{\vx_n} V_{n}(\vrho))$, the claimed bound on the Lipschitz continuity of $F$ follows directly by \Cref{lemma:smooth};
        \item $\BF \defeq \frac{\max_{1 \leq i \leq n} \sqrt{|\cA_i|}}{\zeta^2}$. This follows given that, by \Cref{lemma:Q-V}, $\| F_{i, s} \|_\infty \leq \frac{1}{\zeta^2}$, for any $(i, s) \in \range{n} \times \cS$, in turn implying that $\| F_{i, s} \|_2 \leq \frac{\sqrt{|\cA_i|}}{\zeta^2}$; and
        \item $\alpha \defeq \frac{\sqrt{|\cS| \sum_{i=1}^n |\cA_i|} }{\zeta^2  \|\vrho\|_\infty^2}$. Indeed, for any two joint policies $\vpi, \vpi' \in \Pi$ it holds that
        \begin{equation*}
            \left| \frac{1}{\tilde{d}_{\vrho}^{\vpi}[s]} - \frac{1}{\tilde{d}_{\vrho}^{\vpi'}[s]} \right| \leq \frac{1}{\| \vrho \|_\infty^2} | \tilde{d}_{\vrho}^{\vpi}[s] - \tilde{d}_{\vrho}^{\vpi'}[s] |,
        \end{equation*}
        for any $s \in \cS$. Let us fix the state $s \in \cS$. To bound the term $| \tilde{d}_{\vrho}^{\vpi}[s] - \tilde{d}_{\vrho}^{\vpi'}[s] |$, we consider a fictitious Markov game $\tilde{\cG}$ such that for any player $i \in \range{n}$ the reward is defined so that $\tilde{R}_i(s', \va) = 1$ if $s' = s$, and $0$ otherwise. Then, it follows that the value function takes the form $\tilde{V}_i^{\vpi}(\vrho) = \tilde{d}_{\vrho}^{\vpi}[s]$, for any player $i \in \range{n}$ and joint policy $\vpi \in \Pi$. Thus, the term $ \tilde{d}_{\vrho}^{\vpi}[s] - \tilde{d}_{\vrho}^{\vpi'}[s]$ is equal to $\tilde{V}_1^{\vpi}(\vrho) - \tilde{V}_1^{\vpi_1', \vpi_{-1}}(\vrho) + \dots + \tilde{V}_n^{\vpi_n, \vpi_{-n}'}(\vrho) - \tilde{V}_n^{\vpi'}(\vrho)$, and in turn \Cref{lemma:Valdiff} yields that
        \begin{align*}
            | \tilde{d}_{\vrho}^{\vpi}[s] - \tilde{d}_{\vrho}^{\vpi'}[s] | &\leq \frac{1}{\zeta^2} \sum_{i=1}^n \sum_{s \in \cS} \sum_{a_i \in \cA_i}  | \vx_{i, s}[a_i] - \vx'_{i, s}[a_i]| \\
            &= \frac{1}{\zeta^2} \|\vx - \vx'\|_1. 
        \end{align*}
        The conclusion then follows from the equivalence between the $\ell_1$ and the $\ell_2$ norms.
    \end{itemize}
    As a result, \Cref{theorem:genMVI} along with \Cref{lemma:property} imply that after a sufficiently large number of iterations $T = \poly(n, \sum_{i=1}^n |\cA_i|, |\cS|, 1/\zeta, 1/\|\vrho\|_\infty) \cdot 1/\epsilon^2$, we have computed a point $\vx^{(t)}$ such that
    \begin{equation*}
        \langle \vx^{(t)}, F(\vx^{(t)}) \rangle - \min_{\vxstar \in \cX} \langle \vxstar, F(\vx^{(t)}) \rangle \leq \epsilon.
    \end{equation*}
    Finally, \Cref{lemma:GD} implies that 
    \begin{equation*}
        \sum_{i=1}^n \left( V_i^{\vpi^{(t)}}(\vrho) - V_i^{\dagger, \vpi_{-i}^{(t)}}(\vrho) \right) \geq - G \epsilon,
    \end{equation*}
    where $G \defeq \frac{C_\cG}{\zeta}$, in accordance to \Cref{lemma:GD}. In particular, here we defined $C_\cG$ as
    \begin{equation}
        \label{eq:CG}
         \max_{1 \leq i \leq n} \max_{\vpi_{-i} \in \Pi_{-i}} \left\{ \min_{\vpistar_i \in \Pistar_{i}(\vpi_{-i})} \left\| \frac{d_{\vrho}^{\vpistar_i, \vpi_{-i}}}{\vrho} \right\|_\infty \right\}.
    \end{equation}
    Thus, rescaling $\epsilon' \defeq G \epsilon$ concludes the proof. 
\end{proof}
\section{Illustrative Experiments}
\label{appendix:experiments}

Our main result concerns the behavior of~\eqref{eq:OGD} under a time-varying but non-vanishing learning rate---captured by the term $A(\vx)$ in the update rule of~\eqref{eq:OGD}. In this section, we present some illustrative experiments that juxtapose the performance of the variant we analyze and the standard optimistic gradient descent algorithm under a constant learning rate $\eta > 0$.

Specifically, we conduct experiments on the ratio game~\eqref{eq:ratio}, where $\mat{R} \in \R^{100 \times 120}$ and $\mat{S} \defeq \vec{s} \otimes \vec{1}_{120}$ for some $\vec{s} \in \R^{100}$. Each entry of $\mat{R}$ and $\vec{s}$ are selected uniformly at random from $(0, 1)$. We execute each algorithm for $10^3$ iterations with $\eta \defeq 0.1$. The results for $9$ different random realizations are illustrated in \Cref{fig:experiments}. Overall, we see that the two algorithms attain similar performance, although there are no theoretical guarantees for the performance of~\eqref{eq:OGD} with a constant learning rate.

\begin{figure}[!hbt]
    \centering
    \includegraphics[scale=0.4]{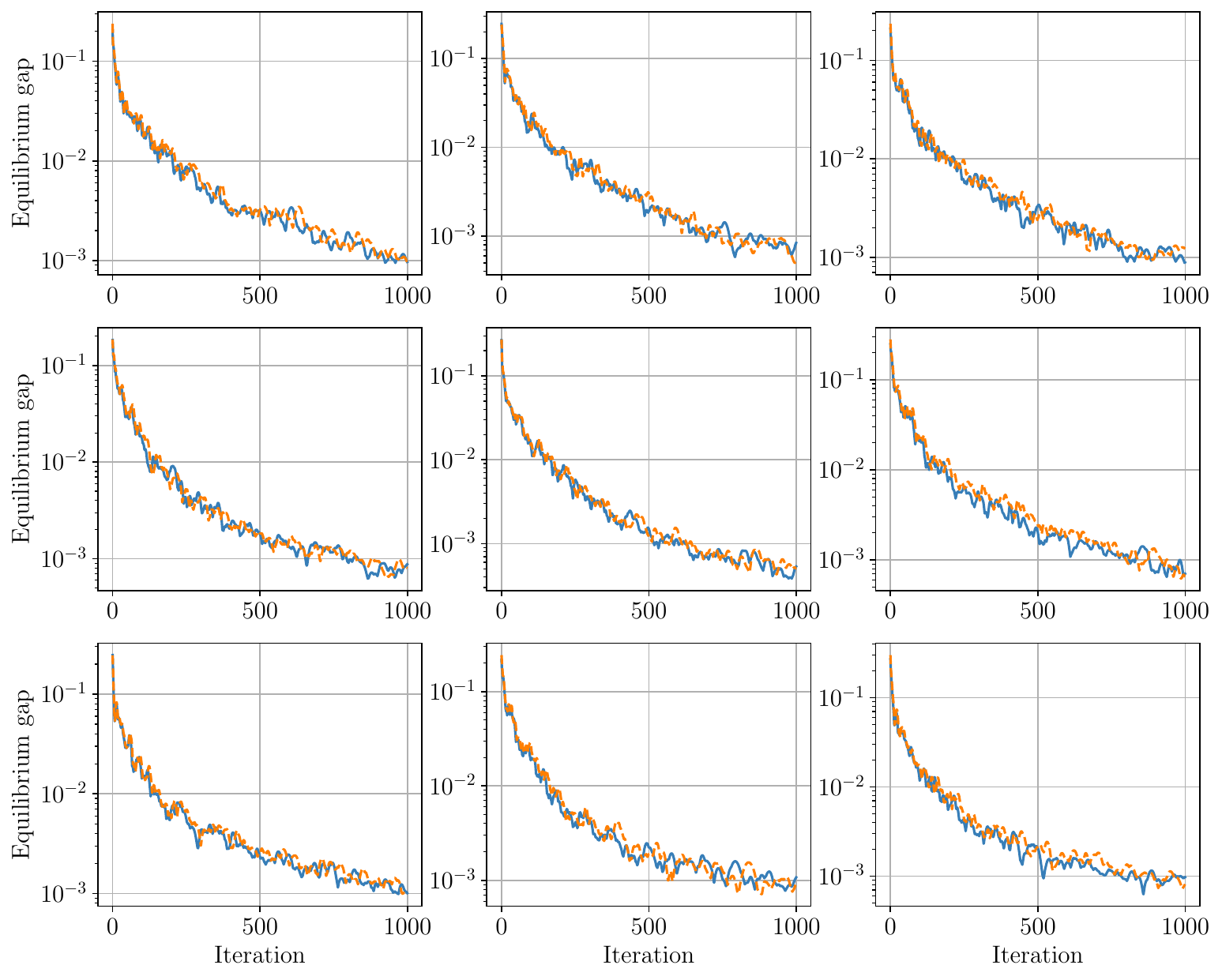}
    \caption{Optimistic gradient descent with a constant learning rate $\eta \defeq 0.1$ (blue curve) versus optimistic gradient descent with a time-varying learning rate per~\eqref{eq:OGD} (orange curve). Each figure corresponds to a separate random ratio game. The equilibrium gap of a joint strategy $\vx \in \cX$ is defined as $\max_{\vxstar \in \cX} \langle \vx - \vxstar, F(\vx) \rangle$. }
    \label{fig:experiments}
\end{figure}

\end{document}